%% file: ucrs2.tex
\documentclass[11pt]{article}
\usepackage[margin=1in]{geometry}
\usepackage{amsfonts}
\usepackage{amsmath,amsthm,amssymb}
\usepackage{latexsym}
\usepackage{epic}
\usepackage{epsfig}
\usepackage{hyperref}
\usepackage{verbatim}
\usepackage[justification=centering]{caption}
\usepackage{enumitem}
\usepackage{color}
\allowdisplaybreaks[1]
\usepackage[numbers]{natbib}
\usepackage{algorithm}
\usepackage{algorithmic}
\usepackage{setspace}

\usepackage{eqparbox}

\usepackage{tikz}
\usetikzlibrary{arrows}

\input{macros}

\title{Matroid Secretary Is Equivalent to Contention Resolution} 

\author{
Shaddin Dughmi\thanks{This work was supported by NSF CAREER Award CCF-1350900 and NSF Grant CCF-2009060.} \\
Department of Computer Science\\
University of Southern California\\
{\tt shaddin@usc.edu}
}

\begin{document}

\maketitle

\begin{abstract}
\input{abstract}
\end{abstract}

\newcommand{\ind}{\mathbf{Ind}}
\renewcommand{\binom}{\mathbf{Binom}}
\newcommand{\spn}{\mathbf{span}}
\newcommand{\rank}{\mathbf{rank}}
\newcommand{\OPT}{\mathbf{OPT}}
\newcommand{\convexhull}{\mathbf{convexhull}}
\newcommand{\val}{\mathbf{val}}
\newcommand{\supp}{\mathbf{supp}}
\newcommand{\imp}{\mathbf{Imp}}
\newcommand{\impl}{\mathbf{Imp}_{\mathtt{lbl}}}

\input{intro}
\input{prelim}
\input{outline}
\input{sec_to_psec}
\input{psec_to_lcr}
\input{lcr_to_cr} 
\input{conclusion}

\newpage

{
\bibliography{ucrs2}
\bibliographystyle{abbrvnat}               
}

\newpage
\appendix
\input{appendix_normalize_discretize}
\newpage

\input{appendix_onlyactive}

\end{document}

%% file: macros.tex
\newtheorem{theorem}{Theorem}[section]

\newtheorem{lemma}[theorem]{Lemma}
\newtheorem{sublemma}[theorem]{Sublemma}

\newtheorem{definition}[theorem]{Definition}

\newtheorem{observation}[theorem]{Observation}
\newtheorem{fact}[theorem]{Fact}
\newtheorem{example}[theorem]{Example}

\newtheorem{prop}[theorem]{Proposition}

\newtheorem{conjecture}{Conjecture}
\newtheorem*{conjecture*}{Conjecture}
\newtheoremstyle{nonindented}{1ex}{1ex}{}{}{\bfseries}{.}{.5em}{}
\newtheoremstyle{indented}{1ex}{1ex}{\itshape\addtolength{\leftskip}{0.6cm}\addtolength{\rightskip}{0.6cm}}{}{\bfseries}{.}{.5em}{}
\theoremstyle{nonindented}
\theoremstyle{indented}
\theoremstyle{plain}

\newenvironment{alg}{\begin{algorithm}\begin{onehalfspace}\begin{algorithmic}[1]}{\end{algorithmic}\end{onehalfspace}\end{algorithm}}

\newcommand{\cross}{\times}

\newcommand{\set}[1]{\left\{ #1 \right\}}
\newcommand{\union}{\cup}
\newcommand{\intersect}{\cap}

\newcommand{\sm}{\setminus}

\renewcommand{\hat}{\widehat}
\renewcommand{\tilde}{\widetilde}
\renewcommand{\bar}{\overline}

\def\max{\qopname\relax n{max}}

\def\Pr{\qopname\relax n{\mathbf{Pr}}}
\def\Ex{\qopname\relax n{\mathbf{E}}}

\newcommand{\RR}{\mathbb{R}}
\newcommand{\RRp}{\RR_+}

\newcommand{\NN}{\mathbb{N}}

\def\C{\mathcal{C}}
\def\D{\mathcal{D}}
\def\E{\mathcal{E}}
\def\F{\mathcal{F}}

\def\I{\mathcal{I}}

\def\L{\mathcal{L}}
\def\M{\mathcal{M}}

\def\P{\mathcal{P}}

\def\W{\mathcal{W}}

\def\eps{\epsilon}
\def\sse{\subseteq}

\newcommand{\eat}[1]{}

\newcommand{\fig}[2][1.0]{
  \begin{center}
    \includegraphics[scale=#1]{#2}
  \end{center}}

\newcommand{\INPUT}{\item[\textbf{Input:}]}

\newcommand{\PARAMETER}{\item[\textbf{Parameter:}]}

\newenvironment{lp*}{\begin{equation*}  \begin{array}{lll}}{\end{array}\end{equation*}}


%% file: abstract.tex
We show that the matroid secretary problem is equivalent to correlated contention resolution in the online random-order model. Specifically, the matroid secretary conjecture is true if and only if every matroid admits an online random-order contention resolution scheme which, given an arbitrary (possibly correlated) prior distribution over subsets of the ground set, matches the balance ratio of the best offline scheme for that distribution up to a constant. We refer to such a scheme as universal.
%
%
Our result indicates that the core challenge of the matroid secretary problem lies in resolving contention for positively correlated inputs, in particular when the positive correlation is benign in as much as offline contention resolution is concerned.

Our result builds on our previous work which establishes one direction of this equivalence, namely that the secretary conjecture implies universal random-order contention resolution, as well as a weak converse, which derives a matroid secretary algorithm from a random-order contention resolution scheme with only partial knowledge of the distribution. It is this weak converse that we strengthen in this paper: We show that  universal random-order contention resolution  for matroids, in the usual setting of a fully known prior distribution, suffices to resolve the matroid secretary conjecture in the affirmative.

Our proof is the composition of three reductions. First, we use  duality arguments to reduce the matroid secretary problem to the matroid prophet secretary problem with arbitrarily correlated distributions. Second, we introduce a bridge problem we term labeled contention resolution --- generalizing classical contention resolution --- to which we reduce the correlated matroid prophet secretary problem, employing structural results from our prior work pertaining to the set of improving elements. Finally, we combine duplication of elements with limiting arguments to reduce labeled contention resolution to classical contention resolution.


%% file: intro.tex
\section{Introduction}

This paper follows in the hallowed TCS tradition of reducing the number of questions without providing any answers. We establish an equivalence between one of the central open problems in online algorithm design, the \emph{matroid secretary conjecture}, and the increasingly rich and fruitful  framework of \emph{contention resolution}. Specifically, we show that the matroid secretary problem admits a constant-competitive algorithm if and only if matroid contention resolution for general (correlated) distributions is approximately as powerful (up to a constant) in the online random-order model as it is in the offline model.  Our result paves the way for application of the many recent advances in contention resolution, and in stochastic decision-making problems more generally, to resolving the conjecture.

The classical (single-choice) secretary problem \cite{Dynkin}, and its many subsequent combinatorial generalizations, capture the essence of online decision making when adversarial datapoints arrive in a non-adversarial order. The paradigmatic such generalization is the \emph{matroid secretary problem}, originally proposed by \citet{matsec}.
Here, elements of a known matroid arrive online in a uniformly random order, each equipped with a nonnegative weight chosen at the outset by an adversary. An algorithm for this problem must decide online whether to accept or reject each element, knowing only the weights of the elements which have arrived thus far, subject to accepting an independent set of the matroid. The goal is to maximize the total weight of accepted elements. 
The \emph{matroid secretary conjecture} of \cite{matsec} postulates the existence of an (online) algorithm for this problem which is constant competitive, as compared to the offline optimal, for all matroids. Though much prior work has designed competitive algorithms for specific classes of matroids, the general conjecture has remained open. 

Recent years have seen an explosion of interest in a variety of online decision-making problems of a similar flavor, albeit distinguished from secretary problems in that the uncertainty in the data is stochastic, with known distribution, rather than adversarial. Such models include variants and generalizations of the classical \emph{prophet inequality}, adaptive stochastic optimization models such as \emph{stochastic probing}, and what is increasingly emerging as the technical core of such problems: \emph{contention resolution}. The offline model of contention resolution was introduced by  Chekuri et al.~\cite{CRS}, motivated by applications to approximation algorithm design. It has since been extended to various online settings (e.g. \cite{OCRS,ROCRS}), and emerged as the basic technical building block of a number of important results for stochastic decision-making problems (see e.g. \cite{OCRS,ROCRS,OCRS_prophet,delegated_probing}).

In contention resolution, elements of a set system --- for our purposes, a matroid --- are each equipped with a single-bit stochastic datapoint indicating whether that element is \emph{active} or \emph{inactive}. The joint distribution of these datapoints, henceforth referred to as the \emph{prior distribution}, is assumed to be known and given.  An algorithm for this problem --- which we often refer to as a \emph{contention resolution scheme (CRS)} --- is tasked with accepting an independent set of active elements with the goal of maximizing the \emph{balance ratio}: the minimum, over all elements, of the ratio of the probability the element is accepted to the probability the element is active. When a CRS achieves a constant balance ratio for a distribution or class of distributions, we simply call it \emph{balanced}. In the original \emph{offline} setting of contention resolution, the algorithm observes all datapoints before choosing which elements to accept. Most pertinent for us is the \emph{online random order} setting: elements and their datapoints arrive in a uniformly random order, and the  algorithm must must decide whether to accept or reject each element, subject to independence, knowing only the activity status of elements which have arrived thus far. 

Most work on contention resolution has restricted attention to product prior distributions: elements are active independently, with given probabilities. Sweeping positive results hold for product priors, for both offline  and online contention resolution on matroids (see \cite{CRS,OCRS}), 
and those results tend to extend to negative correlation between elements. In contrast, it is easy to see that not much is possible in the presence of unrestrained positive correlation, even offline. 
We build on our recent work in \cite{ucrs}, which observed that some forms of positive correlation are relatively ``benign'' for contention resolution, at least in the offline setting.  We characterized  \emph{uncontentious} distributions --- those permitting a balanced offline CRS --- and delineated some of their basic properties. Leveraging this characterization, we then related the matroid secretary conjecture to online contention resolution for these uncontentious distributions, via a pair of complementary reductions.

One of the reductions in \cite{ucrs} is of unambiguous significance, and follows from unsurprising duality arguments: given a competitive algorithm for the secretary problem on a matroid, one can derive an (online) random-order CRS which is balanced for every uncontentious distribution on that matroid. We refer to such an online CRS, which is balanced for all uncontentious (correlated) distributions, as \emph{universal}.

The second reduction in \cite{ucrs} is from the matroid secretary problem to a more restrictive  model for online contention resolution, and therefore falls short of establishing an equivalence between the two problems.  At the center of this reduction is the (random) set of \emph{improving elements} for a weighted matroid, as originally defined by \citet{karger_matroidsampling}: a random \emph{sample} consisting of a constant fraction of the elements is set aside, and an element outside the sample is deemed \emph{improving} if it increases the weighted rank of the sample. 
It is shown in \cite{ucrs} that improving elements, though they may exhibit nontrivial positive correlation, are nonetheless uncontentious --- i.e., they admit a balanced offline~CRS. Achieving such balance online as well, in the random-order model, is then shown to imply the matroid secretary conjecture. The major caveat to this reduction is the following: the prior distribution of improving elements is only partially known when the online CRS is invoked by this reduction. In essence, the reduction requires online contention resolution in a nontraditional, and more restrictive, model of a partially-described prior distribution. 

\subsection*{Results and Technical Approach}


\input{reductions.tex}

This is where the present paper picks up. We restrict attention to matroids, and derive a reduction from the secretary problem to random-order contention resolution with a fully known and given uncontentious distribution.  In doing so, we establish equivalence of the matroid secretary conjecture and universal random-order contention resolution on matroids. A conceptual take-away from our result is that the key challenge of the matroid secretary problem lies in resolving contention for random sets exhibiting positive correlation, in particular when such correlation is ``benign'' for offline contention resolution.

We face a series of technical obstacles, which we isolate by expressing our reduction as the composition of three component reductions. This takes us through two  ``bridge problems''  along the way. The first of these bridge problems is the correlated version of the familiar \emph{matroid prophet secretary} problem of \citet{prophet_secretary_b}, which relaxes the matroid secretary problem by assuming that  weights are drawn from a known  distribution rather than adversarially.\footnote{An alternate, equivalent, description of the prophet secretary problem is as a relaxation of the prophet inequality problem to random order arrivals.} The second bridge problem is a generalization of contention resolution --- in particular on matroids, in the online random-order model ---  which we define and term \emph{labeled contention resolution}. Here, each active element comes with a stochastic label, and  balance is evaluated with respect to element/label pairs rather than merely with respect to elements.   Figure \ref{fig:reductions} summarizes the cycle of reductions between all four problems, which we conclude are all equivalent up to constant factors in their competitive and balance ratios.

Our first component reduction, motivated by the aforementioned caveat to the results of \cite{ucrs}, is from the secretary problem to the prophet secretary problem on matroids. Fairly standard duality arguments allow us to replace the adversarial weight vector in the secretary problem with a stochastic one of known distribution. Modulo some simple normalization and discretization of the weights, at the expense of a constant in the competitive ratio, this yields an instance of the prophet secretary problem. With a stochastic weight vector drawn from a known distribution, we now face a \emph{known} mixture of improving element distributions. Moreover, since it is shown in \cite{ucrs} that uncontentious distributions are closed under mixing, this mixture is still uncontentious. At first glance, it would appear that we have now resolved the caveat of \cite{ucrs}.

Unfortunately, shifting to a stochastic weight vector introduces a new obstacle. With the set of improving elements now correlated with the vector of element weights, balanced contention resolution no longer guarantees extracting a constant fraction of the expected weight of the set. This is because a contention resolution scheme may preferentially accept an improving element when it has low weight, and reject it when it has high weight, while still satisfying the balance requirement in the aggregate. In fact, we show by way of a simple example that egregious instantiations of this phenomenon are not difficult to come by. This motivates our reduction from the matroid prophet secretary problem to  \emph{labeled} contention resolution, also in the online random-order model. By labeling each element with its weight, and requiring balance with respect to element/label pairs, we exclude contention resolution policies which favor low-weight elements.

Our final, and most technically involved, component reduction is from labeled to unlabeled contention resolution, for matroids in the random-order model. Such a reduction would be trivial in the offline setting: by thinking of each (element,label) pair as a distinct parallel copy of the element, we obtain an equivalent instance of unlabeled contention resolution, albeit on a larger matroid.
One might hope for an online version of this reduction, which interleaves inactive element/label pairs amidst the active element/label pairs from the labeled instance. However, we argue at length that such an approach appears unlikely to succeed, for two fundamental reasons. First, we present evidence that not any interleaving will do: we show formally that an arbitrary interleaving  produces a  contention resolution problem  which does not admit a constant balance ratio, ruling out such a reduction if the matroid secretary conjecture were true. In other words, it really is important that both active and inactive elements are ordered randomly in random-order contention resolution, since the semi-random generalization which provides no guarantees on the positions of inactive elements is strictly more difficult (assuming the matroid secretary conjecture).  Second, we argue that natural interleaving approaches fail to produce a uniformly-random sequence of element/label pairs (both active and inactive),  even in an approximate sense. Roughly speaking, the difficulty is thus: natural online reductions from the labeled problem to its unlabeled counterpart must randomly interleave many (inactive) labeled copies of an element early into the sequence, well before the active copy (if any) arrives online. Without knowing the identity of this active copy (if any) in advance, there simply is not enough information, in a statistical distance sense, to approximately simulate a uniformly-random interleaving. We overcome these obstacles by ``blowing up'' the matroid even further, creating a large number of identical duplicates of each label. As the number of duplicates grows large, a random order of element/label pairs converges in distribution to a deterministic order (modulo the equivalence relation between duplicates). The required interleaving of inactive element/label pairs is now essentially deterministic, and in particular approximately invariant --- in a statistical distance sense --- to the identities of active elements and their labels.





\subsection*{Additional Discussion of Related Work}

Contention resolution in the offline setting was formalized by \citet{CRS}, motivated by applications to approximation algorithm design via randomized rounding. For product priors and a given packing set system, \cite{CRS} shows that the optimal offline balance ratio  equals the worst-case \emph{correlation gap}, as defined by \citet{correlation_gap_journal}, of the set system's weighted rank function.
Starting with the work of \citet{OCRS}, contention resolution was extended to online settings and applied to a variety of problems in mechanism design and adaptive stochastic optimization (see also \cite{ROCRS,OCRS_prophet}). Regardless of the set system, balanced contention resolution is obviously only possible for priors that are (approximately) \emph{ex-ante feasible}:  the random set is  feasible on average, in the sense that the per-element marginal probabilities lie in the polytope associated with the set system. One message of the aforementioned prior work is that --- for product priors and many natural set systems such as matroids, knapsacks, and their intersections ---  approximate ex-ante feasibility is also sufficient for balanced contention resolution, whether offline or online in any natural arrival model. Beyond product priors, the difficulty lies with resolving contention in the presence of positive correlation. Without any assumptions on the kind or degree of positive correlation, there exist simple examples of highly contentious yet ex-ante feasible distributions.\footnote{Consider a $1$-uniform matroids with $n$ elements, all of which are active simultaneously with probability $\frac{1}{n}$.} Motivated by the existence of relatively ``benign'' forms of positive correlation, and the connection thereof to the secretary problem, our  work in \cite{ucrs} characterized uncontentious distributions regardless of correlation, and established some of their basic properties.

The (single-choice) secretary problem is due to \citet{Dynkin}. It was subsequently generalized to a uniform matroid constraint by \citet{kleinberg_multisec}, and to a general matroid constraint by \citet{matsec}. A long line of work has designed constant-competitive algorithms for special cases of the matroid secretary problem, and we refer the reader to the semi-recent survey by \citet{dinitz_survey}. The current state-of-the art for general matroids is an $O(\log \log \rank)$-competitive algorithm due to \citet{lachish_loglog}, which was since simplified by \citet{svensson_loglog}. Beyond matroids, the secretary problem with general packing constraints was recently studied by \citet{Rub16}.

Closely related to the secretary problem are the prophet inequality problem and the prophet secretary problem, which analogously admit combinatorial generalizations to matroids and other packing set systems. Whereas a secretary problem features adversarial data (i.e., element weights) arriving online in a random order, a prophet inequality problem features stochastic data (typically assumed to be independent) arriving online in an adversarial order. A prophet secretary problem is a relaxation of both, featuring stochastic data arriving online in a random order. The original (single-choice) prophet inequality is due to Krengel, Sucheston, and Garling \cite{KS77, KS78}, and was generalized to matroids by \citet{matroid_prophet}. Generalizations beyond matroids have also received much study; see for example \cite{prophet_easy,polymatroid_prophet, Rub16}. The (single-choice) prophet secretary problem was introduced by \citet{prophet_secretary_a}, and further studied in \cite{prophet_secretary_c}. Generalizations  to combinatorial constraints, including matroids, were studied by \citet{prophet_secretary_b}. 

One take-away from this paper is that stochastic decision making in the presence of correlation, and in particular positive correlation, is deserving of more attention. Most prior work on aforementioned stochastic decision-making models  restricts attention to product priors, with a few exceptions which we now mention. For contention resolution, the only exception we are aware of is our aforementioned work \cite{ucrs}. The classical (single-choice) prophet inequality was extended to negatively correlated variables by Rinott and Samuel Cahn~\cite{RSC87,SC91}, whereas no nontrivial prophet inequality holds in the presence of unrestrained positive correlation \cite{HK_survey}. The only nontrivial prophet inequalities we are aware of in the presence of positive correlation are from the recent work of \citet{prophet_linear_correlations}, who pose a particular linear model of correlated distributions. 



%% file: reductions.tex
\begin{figure}
\centering
  \begin{tikzpicture}[>=angle 60, very thick, scale=0.75]
    \draw (0,0) node[name=sec, rectangle,draw]{Matroid Secretary};
    \draw (-3,-4) node[name=psec,rectangle,draw]{\begin{tabular}{c} Correlated Matroid Prophet Secretary\end{tabular}};
    \draw (-3,-8) node[name=lcr, rectangle,draw]{\begin{tabular}{c} Universal Labeled Contention Resolution \\(Matroids, Random Order)\end{tabular}};
    \draw (0,-12) node[name=cr,rectangle,draw]{\begin{tabular}{c} Universal Contention Resolution \\(Matroids, Random Order)\end{tabular}};
    \path[->] (psec) edge [bend left=10, above, sloped] node {Section~\ref*{sec:sec_to_psec}} (sec);
    \path[->] (lcr) edge [bend left = 10, above, sloped] node {Section~\ref*{sec:psec_to_lcr}} (psec);
    \path[->] (cr) edge [bend left = 10, above] node[rotate=125]{\ \ Section~\ref*{sec:lcr_to_cr}}  (lcr);
    \path[->] (sec) edge [bend left = 60, above, sloped] node[rotate=-4] {\begin{NoHyper}\ \ \ \ \cite[Theorem 4.1]{ucrs}\end{NoHyper}}  (cr);
  \end{tikzpicture}
  \caption{Reductions between the four relevant problems. An arrow $A \to B$ indicates using an algorithm for problem $A$ to solve problem $B$, i.e., a reduction from $B$ to $A$. All reductions preserve the competitive or balance ratio up to a constant.}
\label{fig:reductions}
\end{figure}
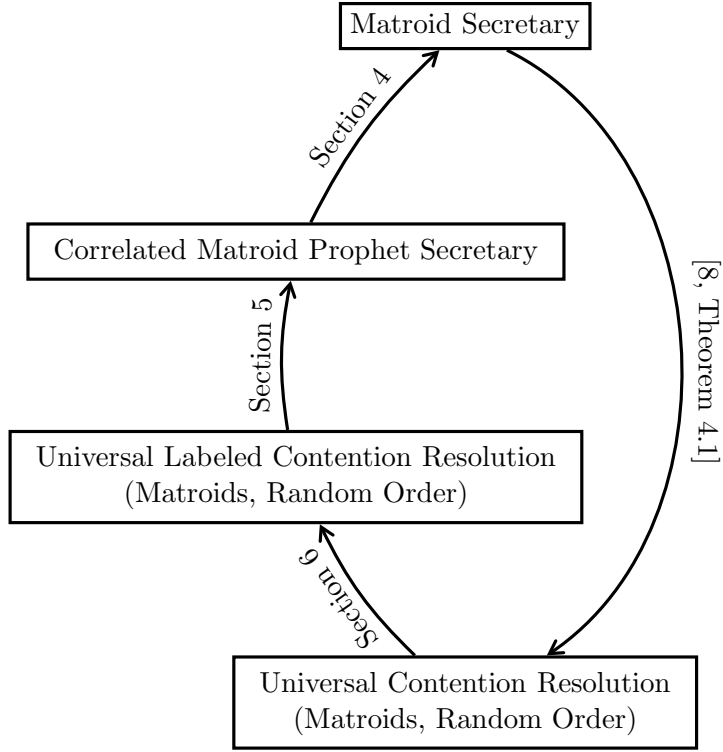

%% file: prelim.tex
\section{Preliminaries}

\subsection{Miscellaneous Notation and Terminology}

We denote the natural numbers by $\NN$, the real numbers by $\RR$, and the nonnegative real numbers by $\RRp$. We also use $[n]$ as shorthand for the set of integers $\set{1,\ldots,n}$.

For a set $A$, we use $\Delta(A)$ to denote the family of distributions over $A$, use $2^A$ to denote the family of subsets of $A$,
and use $A^*$ to denote finite strings with alphabet $A$. When $A$ is finite, we use $a \sim A$ to denote uniformly sampling $a$ from $A$. We also use $A!$ to denote the family of permutations of $A$, where we think of $\pi \in A!$ as a bijection from positions $\set{1,\ldots,|A|}$ to $A$.  When $A$ is equipped with weights $w \in \RR^A$, and $B \sse A$, we use the shorthand $w(B) = \sum_{i \in B} w_i$. For a distribution $\D$ supported on $2^A$, we refer to the vector $x \in [0,1]^A$ of \emph{marginals} of $\D$, where $x_i = \Pr_{B \sim \D}[i \in B]$ is the \emph{marginal probability} of $i$ in $\D$.





\subsection{Matroid Theory Basics}

We use standard definitions from matroid theory; for details see \cite{oxley,welsh}. A matroid $\M=(\E,\I)$ consists of a \emph{ground set} $\E$ of \emph{elements}, and a family $\I \sse 2^\E$ of \emph{independent sets}, satisfying the three \emph{matroid axioms}. A \emph{weighted matroid} $(\M,w)$ consists of a matroid $\M=(\E,\I)$ together with weights $w \in \RR^\E$ on the elements. 
We use $\M | A$ to denote the \emph{restriction} of $\M$ to elements $A \sse \E$.

We use $\rank(\M)$ to denote the \emph{rank} --- i.e. the maximum cardinality of an independent set --- of a matroid $\M$, and $\rank_w(\M)$ to denote the \emph{weighted rank} --- i.e. the maximum weight of an independent set --- of a weighted matroid $(\M,w)$. Overloading notation, we use  $\rank^\M(A)$ to denote the rank of $\M | A$, and $\rank^\M_w(A)$ to denote the weighted rank of $\M|A$ with weights $\set{w_e}_{e \in A}$, though we omit the superscript $\M$ when the matroid is clear from context.  
We also often reference the \emph{matroid polytope} $\P(\M)$ of a matroid $\M=(\E,\I)$, defined as the convex hull of indicator vectors of independent sets.

We restrict attention without loss of generality to matroids with no loops (i.e., each singleton is independent). In parts of this paper, we also restrict attention to weighted matroids where all non-zero weights are distinct. This assumption is made merely to simplify some of our proofs, and --- using standard tie-breaking arguments --- can be shown to be without loss of generality in as much as our results are concerned. Under this assumption, we define $\OPT^\M_w(A)$ as the (unique) maximum-weight independent subset of $A$ of minimum cardinality (excluding zero-weight elements), and we omit the superscript when the matroid is clear from context. We also use $\OPT_w(\M)=\OPT^\M_w(\E)$ as shorthand for the maximum-weight independent set of $\M$ of minimum cardinality.

\subsection{The Matroid Secretary Problem}

In the \emph{matroid secretary problem}, originally defined by \citet{matsec}, there is matroid $\M=(\E,\I)$ with nonnegative weights $\set{w_e}_{e \in \E}$ on the elements. Elements $\E$ arrive online in a uniformly random order $\pi \sim \E!$, and an online algorithm must irrevocably accept or reject an element when it arrives, subject to accepting an independent set of $\M$. Only the matroid $\M$ is given to the algorithm at the outset --- say, as an independence oracle. The weights $w$, on the other hand, are chosen adversarially, and without knowledge of the random order $\pi$.  The elements then arrive online, along with their weights, in the random order $\pi$.


The goal of the online algorithm is to maximize the expected weight of the accepted set of elements. 
Given $c \in [0,1]$, we say that  an algorithm for the secretary problem is \emph{$c$-competitive} for a class of matroids, in the worst-case, if for every matroid $\M$ in that class and every adversarial choice of $w$, the expected weight of the accepted set (over the random order $\pi$ and any internal randomness of the algorithm) is at least a $c$ fraction of the offline optimal --- i.e., at least $ c \cdot \rank_w(\M)$.

The \emph{matroid secretary conjecture}, posed by \citet{matsec}, can be stated as follows.
\begin{conjecture}[\cite{matsec}] \label{conj:matsec}
  There exists an absolute constant $c > 0$ such that the matroid secretary problem admits an (online) algorithm which is $c$-competitive for all matroids.
\end{conjecture}

We note that we are considering the \emph{known matroid} model of the secretary problem, which is the original model defined by~\citet{matsec}. A potentially more challenging  variant, where only the size of the ground set is known at the outset, but the structure of the matroid is revealed online, has also been considered (see e.g.~\cite{gharan_variants}). We are unaware of any evidence of a separation between the two models, and in fact most algorithms in the matroid secretary literature work for both models. Nonetheless, the known setting lends itself best to our reduction.

\subsection{The Matroid Prophet Secretary Problem}
The matroid prophet secretary problem relaxes the matroid secretary problem by assuming that the weights $w \in \RRp^\E$ are drawn from a known \emph{prior distribution} $\mu$, independent of the random order $\pi$, rather than being chosen adversarially. Both $\M$ and $\mu$ are given at the outset, whereas the random order $\pi$ and the realized weight vector $w$ are revealed online as elements arrive. The single-choice prophet secretary problem was introduced by \citet{prophet_secretary_a}, and later studied for matroids and other set systems by \citet{prophet_secretary_b}. To our knowledge, all prior work on the prophet secretary problem has considered independent weights --- i.e., $\mu$ is a product distribution. We make no such assumption here, allowing the weights to be correlated arbitrarily.

Given $c \in [0,1]$, we say that  an algorithm for the secretary problem is \emph{$c$-competitive} for a class of matroids and prior distributions if for every matroid $\M$ and distribution $\mu$ in that class, the expected weight of the accepted set (over the random order $\pi$, the weight vector $w \sim \mu$, and any internal randomness of the algorithm) is at least a $c$ fraction of the expected offline optimal --- i.e., at least $ c \cdot \Ex[\rank_w(\M)]$.

The matroid prophet secretary problem also relaxes the \emph{matroid prophet inequality} problem of \citet{matroid_prophet}, in particular by assuming that the arrival order is uniformly random rather than adversarial. It follows that the competitive ratio of $\frac{1}{2}$ for the matroid prophet inequality from \cite{matroid_prophet} generalizes to the matroid prophet secretary problem when weights are independent. This was improved to $1-\frac{1}{e}$ by \cite{prophet_secretary_b}. No constant is known for the matroid prophet secretary problem with general correlated priors, though one would immediately follow from the matroid secretary conjecture. In fact, along the way to our results we show that the existence of a constant competitive algorithm for the matroid prophet secretary problem, with arbitrary matroids and arbitrary correlated priors, is equivalent to the matroid secretary conjecture.

\subsection{Contention Resolution}

For classical contention resolution, we roughly follow the notation and terminology from \cite{ucrs}.
Let $\M=(\E,\I)$ be a set system. A \emph{contention resolution map (CRM)} $\phi$ for $\M$ is a randomized function from $2^\E$ to $\I$ with the property that $\phi(R) \sse R$ for all $R \sse \E$. Such a map is \emph{$\alpha$-balanced} for a distribution $\rho \in \Delta(2^\E)$  if, for $R \sim \rho$, we have $\Pr[ i \in \phi(R)] \geq \alpha \Pr [ i \in R]$ for all $i \in \E$. Every CRM can be implemented by some algorithm in the \emph{offline model}, where the set $R$ is provided to the algorithm at the outset; when we emphasize this we sometimes say it is an \emph{offline CRM}. If a distribution $\rho \in \Delta(2^\E)$ admits an (offline) $\alpha$-balanced CRM for $\M$, we say $\rho$ is \emph{$\alpha$-uncontentious} for~$\M$. When $R \sim \rho$ and $\rho$ is $\alpha$-uncontentious, we often abuse terminology and also say that the random set $R$ is $\alpha$-uncontentious. We omit reference to $\M$ in these definitions when the set system is clear from context.

The following Theorem characterizes uncontentious distributions for matroids, and the subsequent proposition is an immediate consequence; both are shown in \cite{ucrs}.
\begin{theorem}[\cite{ucrs}]\label{thm:characterize_uncontentious}
  Fix a matroid $\M=(\E,\I)$, and let $\rho \in \Delta(2^\E)$. The following are equivalent for every $\alpha \in [0,1]$.
  \begin{enumerate}[label=(\alph*)]
  \item $\rho$ is $\alpha$-uncontentious (i.e., admits an $\alpha$-balanced offline contention resolution map).
  \item For every weight vector $w \in \RRp^\E$, the following holds for $R\sim \rho$:
\[ \Ex [\rank_w(R)] \geq  \alpha \Ex[w(R)]\]
  \item For every $\F \sse \E$, the following holds for $R \sim \rho$:
    \[ \Ex[\rank(R \intersect \F) ] \geq \alpha \Ex [|R \intersect \F|] \]
  \end{enumerate}
\end{theorem}
\begin{prop}[\cite{ucrs}]\label{prop:uncontentious_mixing}
  Fix a matroid. A mixture of $\alpha$-uncontentious distributions is $\alpha$-uncontentious.
\end{prop}

An \emph{online random-order contention resolution map} (henceforth \emph{RO-CRM} for short) is a CRM $\phi$ which can be implemented as an algorithm in the online random-order model. In the online random-order model,  $\E$ is presented to the algorithm in a uniformly random order $(e_1,\ldots,e_n) \sim \E!$, and at the $i$th step the algorithm learns whether $e_i$ is \emph{active} --- i.e., whether $e_i \in R$ --- and if so must make an irrevocable decision on whether to \emph{accept} $e_i$ --- i.e., include it in the set $\phi(R)$ --- or otherwise \emph{reject} it.

 A \emph{contention resolution scheme (CRS)} $\Phi$ for a set system $\M=(\E,\I)$ and class of distributions $\Delta \sse \Delta(2^\E)$ is an algorithm which takes as input a description of a \emph{prior distribution} $\rho \in \Delta$ and a sample $R \sim \rho$, and outputs $T \in \I$ satisfying $T \sse R$. In effect, $\Phi$ is a collection of contention resolution maps $\phi_\rho$, one for each $\rho \in \Delta$. If each $\phi_\rho$ is $\alpha$-balanced for $\rho$, we say that the $\Phi$ is an $\alpha$-balanced CRS for $\Delta$.  If each $\phi_\rho$ is an RO-CRM, we say that $\Phi$  is an online random order CRS (RO-CRS). Every CRS can be implemented offline, and we say \emph{offline CRS} if we wish to emphasize this.

 In much of the prior work on contention resolution schemes, $\Delta$ was taken to be the class of product distributions with marginals in $\P(\M)$, and each $\rho \in \Delta$ is described completely via its marginals $x \in \P(\M)$. Here, we consider more elaborate classes $\Delta$, most notably $\alpha$-uncontentious distributions for various $\alpha \in [0,1]$. We refer to a balanced CRS for such a class as \emph{universal}.

\begin{definition}
  Fix a set system. For $\beta \leq \alpha \leq 1$, a \emph{$(\beta,\alpha)$-universal CRS} is a CRS which is $\beta$-balanced for the class of $\alpha$-uncontentious distributions. 
\end{definition}

The above definition is only interesting in restricted input models: there always exists an (offline) $(\alpha,\alpha)$-universal CRS for every $\alpha$ and every set system, by definition. Moreover, it is only interesting for $0< \alpha <1$, since the identity CRS is $(\alpha,\alpha)$-balanced otherwise. We will be concerned with the existence of $(\beta,\alpha)$-universal RO-CRS's, for constants $0 < \beta \leq \alpha <1 $, and matroid set systems.





%% file: outline.tex
\section{Overview of Results and Approach}
Our main result is the following.
\begin{theorem}\label{thm:main}
  The following three statements are equivalent
  \begin{enumerate}[label=(\roman*)]
  \item The matroid secretary conjecture (Conjecture~\ref{conj:matsec}). \label{thm:main:conj}
  \item There exists constants $0 < \beta \leq \alpha < 1$ such that every matroid admits a $(\beta,\alpha)$-universal RO-CRS. \label{thm:main:weak}
  \item There exists a constant $0 < c \leq 1$ such that every matroid admits an RO-CRS which is $(c \alpha, \alpha)$-universal, simultaneously for all $\alpha \in [0,1]$. \label{thm:main:strong}
  \end{enumerate}
\end{theorem}



It was shown in \cite[Theorem 4.1]{ucrs} that \ref{thm:main:conj} implies \ref{thm:main:strong}. Moreover, it is easy to see that statement \ref{thm:main:strong} is stronger than \ref{thm:main:weak}. In this paper we show that \ref{thm:main:weak} implies \ref{thm:main:conj}, completing the proof of Theorem~\ref{thm:main}. In particular, we reduce   the matroid secretary conjecture to $(\beta,\alpha)$-universal random-order contention resolution, for arbitrary constants $\beta,\alpha \in (0,1)$.\footnote{A notable, and perhaps surprising, consequence of Theorem~\ref{thm:main} the existence of an $(\Omega(1),\alpha)$-universal RO-CRS on matroids for some  $\alpha \in (0,1)$ implies the same for all other $\alpha' \in (0,1)$. Even more so, it implies the existence of the strong form of universal RO-CRS in \ref{thm:main:strong}.} 
%
We emphasize that, unlike in \cite[Section 5]{ucrs}, we reduce the matroid secretary problem to random-order contention resolution in the traditional setting of a known and given prior distribution.

First, we introduce a ``bridge problem'' which we term \emph{labeled contention resolution}, generalizing classical contention resolution.

\subsection{Labeled Contention Resolution}\label{sec:lcr}
Labeled contention resolution generalizes (classical) contention resolution to a setting where each active element arrives with a label, and a scheme is $\alpha$-balanced if each (element,label) pair is accepted with probability at least $\alpha$-times the probability that the element is active with that label. More formally, let $\M=(\E,\I)$ be a set system, and let $\L$ be a finite set of \emph{labels}. A \emph{labeled set} for $(\M, \L)$ is a pair $(R,L)$ where $R \sse \E$ and  $L: R \to \L$ is an \emph{labeling} of $R$ with $\L$.  A \emph{labeled contention resolution map (LCRM)} $\phi$ for $(\M,\L)$ takes as input such a labeled set $(R,L)$, where $R$ is again referred to as the set of \emph{active} elements, and outputs $T \in \I$ with the property $T \sse R$. Such an LCRM is \emph{$\alpha$-balanced} for a distribution $\rho$ over labeled sets for $(\M,\L)$ if, when the input $(R,L)$ is drawn from $\rho$, we have $\Pr[ e \in \phi(R,L) \land L(e) = \ell] \geq \alpha \Pr [ e \in R \land L(e)=\ell]$ for every $e \in \E$ and $\ell \in \L$. When an (offline) $\alpha$-balanced LCRM exists for a distribution $\rho$ over labeled sets, we again say that $\rho$ is \emph{$\alpha$-uncontentious} for $\M$. When $(R,L) \sim \rho$ and $\rho$ is $\alpha$-uncontentious, we often abuse terminology and also say that the random labeled set $(R,L)$ is $\alpha$-uncontentious. We omit reference to $\M$ and/or $\L$ when they are clear from context.

In the online random order setting, elements of $\E$ arrive in a uniformly random order $(e_1,\ldots,e_n)$, and at the $i$th step the algorithm learns whether $e_i$ is \emph{active} --- i.e., whether $e_i \in R$ --- and if so the algorithm also learns its label $L(e_i)$. The algorithm must then make an irrevocable decision on whether to accept $e_i$.

Remaining notions and terms from unlabeled contention resolution generalize naturally to the labeled setting:  A \emph{labeled contention resolutions scheme (LCRS)} $\Phi$ for set system $\M$ takes as input a description of a distribution $\rho$ over labeled sets for $\M$ and some finite set $\L$ of labels, and implements an LCRM $\phi_\rho$ for $(\M,\L)$. As before, an LCRS $\Phi$ may offline or online, and is $\alpha$-balanced for a class of distributions if, for $\rho$ in that class, $\phi_\rho$ is $\alpha$-balanced for $\rho$. We focus on \emph{$(\beta,\alpha)$-universal RO-LCRSs}: those which are $\beta$-balanced for all $\alpha$-uncontentious distributions over labeled sets (for every finite set of labels), in the online random order model.


Note that classical contention resolution is the special case of labeled contention resolution in which each element of the ground set is associated with a single label. We also note that labeled contention resolution offers little beyond classical contention resolution in the offline model for matroids: if we think of labeled copies of an element as parallel elements in a new matroid, we obtain an equivalent unlabeled contention resolution problem.\footnote{More generally, this is also the case for any family of set systems closed under duplication of elements.} Formally, for a matroid $\M=(\E,\I)$ and set $\L$ of labels, we define their ``tensor product'' $\M\otimes \L = (\E \times \L, \I \otimes \L)$, where $\I \otimes \L$ includes $S \odot L = \set{ (e,L(e)): e \in S}$ for each $S \in \I$ and each $L: S \to \L$. It is easy to verify that $\M \otimes \L$ is a matroid: each element of $\M$ was just replaced with $|\L|$ parallel elements, one for each label.  In the offline setting, a labeled contention resolution problem on $\M$ and $\L$ is equivalent to an unlabeled one on $\M \otimes \L$. In particular, we can think of a  labeled set $(R,L)$ for $\M$ and $\L$ as an (unlabeled)  set $R \odot L = \set{(e,L(e)): e \in R}$ for $\M \otimes \L$. It follows that a random labeled set $(R,L)$ is $\alpha$-ucontentious (in the labeled sense, for $\M$ and $\L$)  if and only if the corresponding unlabeled set $R \odot L$ is $\alpha$-uncontentious (in the unlabeled sense, for $M \otimes \L$). Given this equivalence, the following labeled analogue of Proposition~\ref{prop:uncontentious_mixing}, which will be useful in Section~\ref{sec:psec_to_lcr}, is immediate.
\begin{prop}\label{prop:labeled_uncontentious_mixing}
  Fix a matroid and a set of labels. A mixture of $\alpha$-uncontentious distributions over labeled sets is $\alpha$-uncontentious.
\end{prop}

Our main concern will be labeled contention resolution in the online random order model. Unlike in the offline model, the reduction from the labeled to the unlabeled problem is nontrivial, as will be shown in Section~\ref{sec:lcr_to_cr}.\footnote{Though not a concern of this paper, the relationship between the labeled and unlabeled problems is interesting to contemplate in other online order models. In the adversarial order model, it is not too hard to see that the two problems are again equivalent. In the free order model, however, no such equivalence is immediately obvious.}

\subsection{Proof Outline}

Our proof is the composition of three reductions, one from the matroid secretary problem to the (correlated) matroid prophet secretary problem, one from the matroid prophet secretary problem to  universal labeled contention resolution, and finally one from  labeled to unlabeled contention resolution, all in the online random order model. Theorem \ref{thm:main} is a consequence of the following three lemmas, combined with the reverse reduction in \cite[Theorem 4.1]{ucrs}.

\begin{lemma}\label{lem:sec_to_psec}
  Fix a constant $c \in [0,1]$. If there is a $c$-competitive algorithm for the matroid prophet secretary problem with finitely-supported arbitrarily-correlated priors, then there is a $\frac{c}{4096}$-competitive algorithm for the matroid secretary problem. 
\end{lemma}

\begin{lemma}\label{lem:psec_to_lcr}
 Fix constants $0 \leq \beta \leq \alpha \leq 1$.   If there is a $(\beta,\alpha)$-universal RO-LCRS for a matroid $\M$, then there is a $\beta (1-\alpha)$-competitive algorithm for the  matroid prophet secretary problem on $\M$ with finitely-supported arbitrarily-correlated priors.
\end{lemma}

\begin{lemma}\label{lem:lcr_to_cr}
 Fix constants $0 \leq \beta \leq \alpha \leq 1$.  If every matroid admits a $(\beta,\alpha)$-universal RO-CRS, then for each $\tilde{\beta} < \beta$, every matroid  admits a $(\tilde{\beta},\alpha)$-universal RO-LCRS.
\end{lemma}
\noindent We prove Lemmas~\ref{lem:sec_to_psec}, \ref{lem:psec_to_lcr}, and \ref{lem:lcr_to_cr} in Sections ~\ref{sec:sec_to_psec}, \ref{sec:psec_to_lcr}, and \ref{sec:lcr_to_cr}, respectively. Recall Figure~\ref{fig:reductions}.


%% file: sec_to_psec.tex
\section{Reducing  Secretary to Prophet Secretary}
\label{sec:sec_to_psec}
We now reduce the matroid secretary problem  to the matroid prophet secretary problem with a finitely-supported, arbitrarily-correlated prior distribution on weight vectors. Our reduction loses a constant factor in the competitive ratio.

First, we observe that we can restrict attention to instances of the matroid secretary problem which are \emph{normalized}, in that the offline optimal value is roughly $1$, and \emph{discretized}, in that weights are contained in a known finite set. The following Sublemma is shown using standard arguments, and its proof is therefore deferred to Appendix~\ref{app:normalize_discretize}. We note that we make no attempt to optimize the constants here.

\begin{sublemma}
\label{slem:normalize_discretize}
  The matroid secretary problem reduces, at a cost of a factor of $256$ in the competitive ratio, to its special case where the matroid $\M=(\E,\I)$ and weights $w$ are guaranteed to satisfy the following:
  \begin{itemize}
  \item Normalized: $\rank_w(\M) \in \left[\frac{1}{16},1 \right]$. 
  \item Discretized: The weight $w_e$ of each element $e \in \E$ is either zero, or is an integer power of $2$ contained in $\left[\frac{1}{256\ \rank(\M)},1\right]$.
  \end{itemize} 
\end{sublemma}

We now fix the matroid $\M=(\E,\I)$, and reduce the normalized and discretized matroid secretary problem on $\M$, in the sense of Sublemma~\ref{slem:normalize_discretize}, to the prophet secretary problem on the same matroid $\M$, losing a constant factor in the reduction. To keep the proof generic, we use $a=\frac{1}{16}$ to denote the (known) constant such that offline optimal value is guaranteed to lie in $[a,1]$, and use $W=\set{0} \union \set{2^{-i} : i \in \NN, i \leq \log_2(256\ \rank(\M)) }$ to denote the (known) finite set of \emph{permissible weights} for $\M$. We also use $\W = \set{w \in W^\E : \rank_w(\M) \in [a,1]}$ to denote the (known) finite set of \emph{permissible weight vectors} for $\M$, yielding a normalized and discretized instance.


Our reduction invokes minimax duality  to replace the adversarially-chosen weight vector $w$, as in the secretary problem, with a weight vector drawn from a known and arbitrarily-correlated distribution $\mu$, as in the prophet secretary problem. Discretization is needed so that we can invoke the minimax theorem for finite games. However, straightforward application the minimax theorem produces a variant of the prophet secretary problem where the goal is to maximize the expected ratio between the online and offline optimal values, rather than the (usual) goal of maximizing the ratio of the two expectations. Normalization serves to obviate the distinction between these two~goals.

An algorithm $A$ for normalized and discretized secretary problem on $\M$ maps a permissible weight vector $w \in \W$ and an order $\pi \in \E!$ on the elements to an independent set $A(\pi,w) \in \I$. When $A$ is deterministic, we can think of it as a function from $\W \cross \E!$ to $\I$. Since $\W$, $\E!$, and $\I$ are all finite sets, there are finitely many such functions that are computable online. A randomized algorithm can be thought of as simply a distribution over these functions. For an algorithm $A$, be it deterministic or randomized, we use $\val(A,w) = \Ex_\pi [ w(A(\pi,w)) ]$ to denote the expected weight of the independent set chosen by algorithm $A$ for weight vector $w$, where expectation is over the uniformly random order $\pi \sim \E!$. Note that $\val(A,w)$ is a random variable when $A$ is randomized.

Consider the following finite two-player zero-sum game played between an algorithm player and an adversary. The pure strategies of the algorithm player are  deterministic algorithms for the secretary problem on $\M$, which we think of as functions from $\W \cross \E!$ to $\I$, and mixed strategies are naturally randomized algorithms. The pure strategies for the adversary are the permissible weight vectors $\W$. The algorithm player's utility if he plays a deterministic algorithm $A$ and the adversary plays $w$ is simply the competitive ratio of $A$ on $w$, given by $\frac{\val(A,w)}{\rank_w(\M)}$.

For a randomized algorithm $A$ for the secretary problem, its competitive ratio on a weight vector $w$ is given by$\frac{\Ex[\val(A,w)]}{\rank_w(\M)} = \Ex\left[\frac{\val(A,w)}{\rank_w(\M)}\right]$, where expectation is over any internal randomness in $A$. The worst-case competitive ratio of $A$ is at least $d$ if
\begin{equation}
  \label{eq:maxmin}
  \forall w \in \W:  \Ex_A\left[\frac{\val(A,w)}{\rank_w(\M)}\right]  \geq d
\end{equation}
Inequality \eqref{eq:maxmin} can be equivalently interpreted as follows: if the algorithm player moves first by playing mixed strategy $A$, he guarantees an expected utility of at least $d$ regardless of the response $w$ of the adversary. By the minimax theorem for finite two-player zero-sum games, and through the associated dual pair of linear programs,  the design of an algorithm $A$ satisfying Inequality~\eqref{eq:maxmin} reduces to the following (dual) problem faced by an algorithm player who moves second: for each $\mu \in \Delta(\W)$ (a mixed strategy of the adversary), design an algorithm $B=B(\mu)$ for the secretary problem on $\M$ which satisfies:
\begin{equation}
  \label{eq:minmax}
  \Ex_B \Ex_{w \sim \mu}   \left[\frac{\val(B,w)}{\rank_w(\M)}\right]  \geq d
\end{equation}
We note that our minimax reduction is not necessarily efficient, as both players in our zero-sum game have exponentially many strategies in the size of the ground set of the matroid. An efficient reduction is not necessary, however, for our (information theoretic) result. We also note that there is no benefit to randomization in $B$ when computational efficiency is not a concern: a randomized algorithm $B$ satisfying inequality~\eqref{eq:minmax} can be derandomized, albeit perhaps inefficiently, by appropriately choosing a deterministic algorithm in its support. Nevertheless, we permit randomization in $B$ for our reduction to be as general as possible.\footnote{This is convenient since the reduction from the prophet secretary problem to contention resolution in Sections~\ref{sec:psec_to_lcr} and~\ref{sec:lcr_to_cr} will, in general, produce a randomized algorithm, as contention resolution schemes are typically randomized.}  

Finally, we claim that a $c$-competitive algorithm $B$ for the prophet secretary problem on $\M$ and $\mu$ satisfies inequality~\eqref{eq:minmax} with $d=a \cdot c$.  By definition, the assumption that $B$ is a $c$-competitive prophet secretary algorithm for $\M$ and $\mu$ can be written as
\begin{align*}
  \frac{ \Ex_B \Ex_{w \sim \mu}  [\val(B,w)]}{\Ex_{w \sim \mu} [ \rank_w(\M)]} \geq c.
\end{align*}
It follows that
\begin{align*}
  \Ex_B \Ex_{w \sim \mu} \left[\frac{\val(B,w)}{\rank_w(\M)}\right] 
                                                              &\geq  \Ex_B \Ex_{w \sim \mu} \left[\val(B,w)\right] &\mbox{($\rank_{w}(\M) \leq 1$ for all $w \in \W$)} \\
                                                              &=a \cdot   \frac{ \Ex_B\Ex_{w \sim \mu}  \left[\val(B,w)\right]}{a} \\
                                                              &\geq a \cdot \frac{ \Ex_B \Ex_{w \sim \mu}  \left[\val(B,w)\right]}{\Ex_{w \sim \mu}[ \rank_w(\M)]} &\mbox{($\rank_{w}(\M) \geq a$ for all $w \in \W$)} \\
                                                              &\geq a \cdot c
\end{align*}
Since our reduction lost a factor of $256$ in the normalization and discretization step (Sublemma~\ref{slem:normalize_discretize}), and a factor of $1/a=16$ due to the discrepancy between the objective of the matroid prophet secretary problem and the dual of the matroid secretary problem, this completes the proof of Lemma~\ref{lem:sec_to_psec} with the claimed loss in the competitive ratio of $256 \times 16 = 4096$.


%% file: psec_to_lcr.tex
\section{Reducing  Prophet Secretary to Labeled Contention Resolution}
\label{sec:psec_to_lcr}

Recall that in \cite{ucrs}, the matroid secretary problem is ``reduced'', with a major caveat, to  random-order contention resolution for the set of improving elements.  The random set of improving elements, adapted from the original definition of \citet{karger_matroidsampling}, is defined next. In our definition for improving elements, and in this section generally, we assume the non-zero entries of a matroid weight vector are distinct; this is without loss of generality by standard tie-breaking arguments, and serves to simplify our definitions and proofs.

\begin{definition}[See \cite{karger_matroidsampling,ucrs}]\label{def:improving}
 Let $\M=(\E,\I)$ be a matroid, let $p \in (0,1)$ be a parameter, and let $w \in \RRp^\E$ be a weight vector. The random set $R$ of \emph{improving elements for $(\M,p,w)$} is sampled as follows: Let $S \sse \E$ include each element $e \in \E$ independently with probability $p$, and let $R = \set{ e \in \E : \rank^\M_w(S \union e) > \rank^\M_w(S)}$.\footnote{Equivalently, $R$ is the set of elements in $\E \sm S$ which are not spanned by higher weight elements in  $S$. Another equivalent definition is $R= \set{i \in \E \sm S: i \in \OPT^\M_w(S \union i)}$.} We say $e \in R$ \emph{improves} $S$, in the sense that adding $e$ to $S$ improves its weighted rank. We use $\imp(\M,p,w)$ to denote the distribution of $R$.
\end{definition}

Key to the ``reduction'' in \cite{ucrs} are the following two properties of the set of improving elements.

\begin{fact}[\cite{ucrs}]\label{fact:improving_approximate}
  Let $\M=(\E,\I)$ be a matroid, let $p \in (0,1)$, and let $w \in \RRp^\E$ be a weight vector. The set of improving elements for $(\M,p,w)$ holds a $1-p$ fraction of the weighted rank of $\M$ in expectation. Formally: \[\Ex_{R \sim \imp(\M,p,w)} [ w(R) ]  \geq (1-p) \rank_w(\M).\]
\end{fact}

\begin{theorem}[\cite{ucrs}]\label{thm:improving_uncontentious}
  Let $\M=(\E,\I)$ be a matroid, let $p \in (0,1)$, and let $w \in \RRp^\E$ be a weight vector. The distribution $\imp(\M,p,w)$ is $p$-uncontentious for $\M$.
\end{theorem}
\noindent Fact~\ref{fact:improving_approximate} follows easily from the observation that each element in $\OPT_w(\M)$ is improving with probability $1-p$.  Theorem~\ref{thm:improving_uncontentious}, on the other hand, is nontrivial, and we refer the reader to \cite{ucrs} for its proof.

Consider the following ``reduction'', outlined in \cite{ucrs}, from the secretary problem on an $n$-element matroid $\M=(\E,\I)$, and (a-priori unknown) weights $w$, to online contention resolution: Observe the weights of the first $k \sim \binom(n,p)$ elements $S$ arriving online, then resolve contention for the set of elements $R \sse \E \sm S$ which improve $S$ as they arrive online.\footnote{Note that membership in $R$ can be determined online, as needed.} When $p \in (0,1)$ is a constant, $R$ follows an $\Omega(1)$-uncontentious distribution (Theorem~\ref{thm:improving_uncontentious}), and holds a constant fraction of the optimal value (Fact~\ref{fact:improving_approximate}).  Therefore, it suffices to resolve contention online for $R$ almost as well (up to a constant in the balance ratio) as is possible offline. Since $\E \sm S \supseteq R$ arrive in uniformly random order after $S$, and we can ``interleave'' $S$ among them to create a uniformly random order on $\E$, universal contention resolution in the random order model suffices. The important caveat to this ``reduction'' of \cite{ucrs} is that the distribution $\imp(\M,p,w)$ of $R$, being a function of the unknown and adversarial weight vector $w$, is unknown to the contention resolution scheme. This is a departure from the traditional notion of contention resolution, involving a known and given prior distribution.

In this section, we overcome this caveat by instead reducing from the prophet secretary problem, where $w$ is drawn from a known prior distribution $\mu$. Proposition \ref{prop:uncontentious_mixing} implies that the set of improving elements $R(w) \sim \imp(\M,p,w)$ is still $p$-uncontentious when $w$ is random. This, however, introduces additional difficulties: contention resolution with a constant balance ratio no longer recovers a constant fraction of the weighted rank when $w$ and $R$ are correlated, as illustrated by the following example.
\begin{figure}
  \fig[0.3]{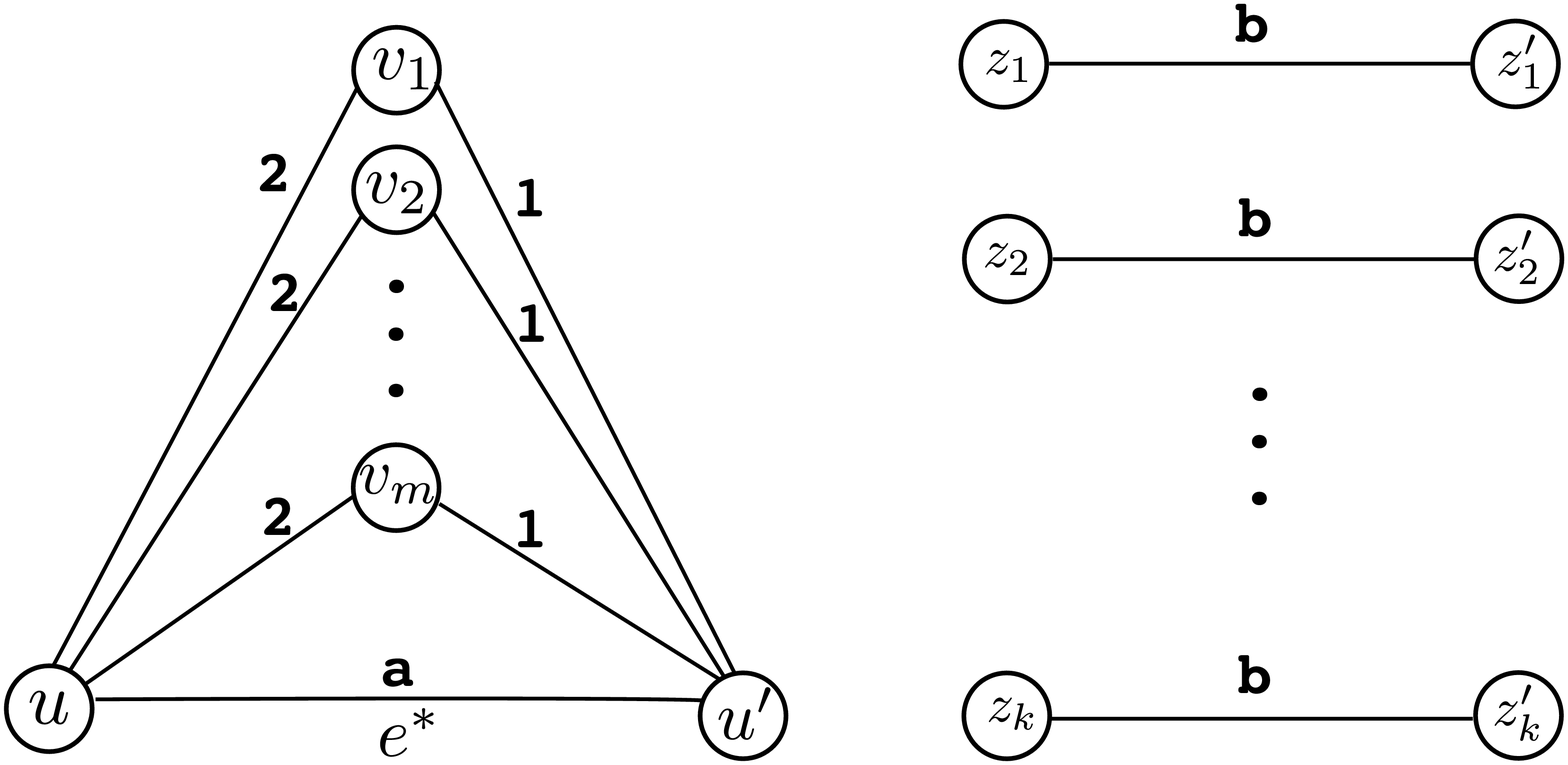}
  \caption{Modified Hat Example. This graphical matroid is truncated to rank $m$.} 
  \label{fig:hat}
\end{figure}

\begin{example}
  Consider the truncated graphical matroid in Figure~\ref{fig:hat}, with the weights labeling the edges and $k > > 2m$. We can guarantee that  weights are distinct by introducing small perturbations.  The graph on the left is the classical  ``hat example'' often employed in the literature on the matroid secretary problem.  We take the disjoint union of the hat example with the free matroid on $k$ elements (represented by the $k$ isolated edges on the right), and truncate the resulting matroid to a rank of $m$.  We fix the sampling parameter $p=\frac{1}{2}$, and examine the set of improving elements for two settings of the weights $a$ and $b$.

  For the first scenario, let $a=4$ and $b=3$. With high probability as $k$ grows large, the set $R$ of improving elements does not include any of the ``hat'' edges with weights $1$ or $2$. Moreover, $\Pr[ e^* \in R] = \frac{1}{2}$. The following simple scheme is $(\frac{1}{2}- o(1))$-balanced: Discard the edges $\set{(z_i,z'_i): i \in [k]}$ with probability $\frac{1}{2}$ (and otherwise discard nothing),\footnote{Discarding these edges serves solely to guarantee balance for the ``hat edges'', in the low probability event that any of the hat edges are improving.} then run greedy random-order contention resolution on the remaining edges. 

  For the second scenario we let $a=\infty$ (or a very large constant) and $b=0$. Setting $b=0$ effectively takes the edges $\set{(z_i,z'_i): i \in [k]}$   ``out of the running'', leaving only the hat example. The set $R$ of improving elements, though $\frac{1}{2}$-uncontentious, is now less amenable to greedy contention resolution: when $e^* \in R$, there are typically many ``hats'' in $R$ as well: for a   constant fraction of $i \in [m]$, both edges $(u, v_i)$ and $(v_i, u')$ are in $R$. It follows that the above-described discard-then-greedy scheme is no longer $\Omega(1)$-balanced. In particular, it selects $e^*$ with probability $O(\frac{1}{m})= o(1)$, despite the fact that $\Pr[e^* \in R] = \frac{1}{2}$. A slightly more involved contention resolution scheme is needed for a constant balance ratio.

  Suppose we randomize between the above scenarios, with each scenario equally likely. Let $R$ be the resulting set of improving elements, and note that $R$ is $\frac{1}{2}$-uncontentious. It is easy to verify that the discard-then-greedy scheme is $(\frac{1}{4} - o(1))$-balanced here. However, $e^*$ is accepted with probability $\frac{1}{2} - o(1)$ when it is active with weight $4$ (in the first scenario), but with probability $o(1)$ when it is active with weight $\infty$ (in the second scenario). Therefore, the discard-then-greedy scheme does not recover a constant fraction of the expected weighted rank of $R$, despite being $\Omega(1)$-balanced.

  A similar situation arises for any nontrivial randomization between the two scenarios, even if we make the second scenario exceedingly unlikely.
\end{example}
This example suggests that we must constrain contention resolution to not ``favor'' improving elements that have low weight. We accomplish this by labeling each improving element with its weight, and requiring contention resolution in the (stronger) labeled sense. We use the following labeled notion of improving elements:

\begin{definition}
  Let $\M=(\E,\I)$ be a matroid, let $p \in (0,1)$ be a parameter, and let $w \in \RRp^\E$ be a weight vector. The random \emph{labeled set of improving elements for $(\M,p,w)$} is the pair $(R,L)$, where $R \sim \imp(\M,p,w)$ is the (random) set of improving elements, and $L:R \to \RRp$  is the labeling with $L(e) = w(e)$ for all $e \in R$. We use $\impl(\M,p,w)$ to denote the distribution of the labeled set~$(R,L)$.
\end{definition}
When $w$ is fixed, each element $e$ is associated with a single label $w(e)$, so labeled contention resolution for $(R,L) \sim \impl(\M,p,w)$ is equivalent to unlabeled contention resolution for $R \sim \imp(\M,p,w)$, and by Theorem~\ref{thm:improving_uncontentious} it follows that $\impl(\M,p,w)$ is $p$-uncontentious in the labeled sense. When $w$ is a drawn from a known prior $\mu$ with finite support, the labeled set of improving elements $(R,L)$ is drawn from a mixture of the $p$-uncontentious distributions $\impl(\M,p,w)$, for the finitely-many realizations of $w \in \supp(\mu)$. When  $w \sim \mu$ and  $(R,L) \sim \impl(\M,p,w)$, we refer to $(R,L)$ as the \emph{labeled set of improving elements for $(\M,p,\mu)$}, and denote its distribution by $\impl(\M,p,\mu)$.  The following is then a direct consequence of Proposition~\ref{prop:labeled_uncontentious_mixing}.

\begin{sublemma}[Follows from Theorem~\ref{thm:improving_uncontentious} and Proposition~\ref{prop:labeled_uncontentious_mixing}]\label{slem:labeled_improving_uncontentious}
  Let $\M=(\E,\I)$ be a matroid, let $p \in (0,1)$, and let $\mu \in \Delta(\RRp^\E)$ be a distribution over weight vectors with finite support. The distribution $\impl(\M,p,\mu)$ is $p$-uncontentious (in the labeled sense) for $\M$.
\end{sublemma}

Fixing matroid $\M=(\E,\I)$ and $p \in (0,1)$, we reduce the prophet secretary problem on $\M$  to  $(\beta,\alpha)$-universal random-order labeled contention resolution with $\alpha=p$.  
The reduction is shown in Algorithm~\ref{alg:psec_to_lcr} for the prophet secretary problem on $\M$, which  takes as an offline input a prior $\mu$ on weight vectors, and as its online inputs a sequence of weighted elements of $\M$. 
We assume that the online inputs to the Algorithm are distributed as specified in the prophet secretary problem, namely with $w \sim \mu$ and $\pi \sim \E!$ drawn independently, and analyze the algorithm's competitive ratio. In particular, we will see that the algorithm achieves its competitive ratio by resolving contention, in the random order model, for a labeled set drawn from the $p$-uncontentious distribution~$\impl(\M,p,\mu)$.

\begin{alg}
  \PARAMETER Matroid $\M=(\E,\I)$ with $n$ elements.
  \PARAMETER $(\beta,\alpha)$-universal RO-LCRS $\Phi$ for matroid $\M$
  \INPUT Finitely-supported prior distribution $\mu \in \Delta(\RRp^\E)$.
  \INPUT Online string $(e_1,w(e_1)), \ldots, (e_n,w(e_n))$, where $\pi=(e_1,\ldots,e_n) \in \E!$, and~$w \in \supp(\mu)$.
  \STATE Let $p=\alpha$
  \STATE Instantiate $\Phi$ with prior distribution $\impl(\M,p,\mu)$, and let $\phi$ denote the resulting RO-LCRM for matroid $\M$ and finite set of labels $\L=\set{ w'(e) : w' \in \supp(\mu), e \in \E}$. \label{step:lcrm}
  \STATE Sample $k \sim \binom(n,p)$.
  \STATE Observe first $k$  online inputs $(e_1,w(e_1)),\ldots (e_k, w(e_k))$ without accepting any.
  \STATE Let $S= \set{e_1,\ldots,e_k}$.
  \STATE Let $i=1$ and $j=k+1$
  \COMMENT{Indexes elements $e_i \in S$ and $e_j \in \E \sm S$}
  \WHILE[While not all elements in $\E$ have been fed to $\phi$]{$i\leq k$ or $j \leq n$}
  \STATE Flip a biased coin with heads probability $\frac{n-(j-1)}{n-(j-1) + k-(i-1)}$ \label{step:coin}
  \IF[Feed next element in $\E \sm S$ to $\phi$]{Coin came up heads}
  \STATE Read the next online input $(e_j,w(e_j))$
  \IF{$\rank^\M_w(S \union e_j) > \rank^\M_w(S)$ (i.e., $e_j$ improves $S$)} \label{step:improving}
  \STATE Feed $e_j$ as {\bf active} to $\phi$, with label $w(e_j)$. {\bf Accept} $e_j$ if $\phi$ accepts it, otherwise {\bf Reject}~$e_j$.\label{step:feed_active}
  \ELSE
  \STATE Feed $e_j$ as {\bf inactive} to $\phi$.
  \ENDIF
  \STATE Increment $j$
  \ELSE[Coin came up tails. Feed an element from $S$ to $\phi$] 
  \STATE Feed $e_i$  as {\bf inactive} to $\phi$
  \STATE Increment $i$
  \ENDIF
  \ENDWHILE
  \caption{Reduction from matroid prophet secretary to labeled contention resolution.}
  \label{alg:psec_to_lcr}
\end{alg}

Let $R$ denote the elements improving $S$, as determined in Step~\eqref{step:improving}, and let $L(e) = w(e)$ be the label of $e \in R$ determined in Step~\eqref{step:feed_active}. We also let $\pi'$ denote the list of elements (whether active or inactive) fed to $\phi$ by Algorithm~\ref{alg:psec_to_lcr}, in that order. First, we show that the inputs to $\phi$ are as stipulated in random-order contention resolution for $\impl(\M,p,\mu)$, and that $\phi$ is $\beta$-balanced for that distribution.

\begin{sublemma}\label{slem:phi_input}
  The labeled set $(R,L)$ follows the distribution $\impl(\M,p,\mu)$. Moreover, $\pi'$ is a uniformly random order on $\E$ independent of $(R,L)$.
\end{sublemma}
\begin{proof}
  Since $\pi$ is a uniformly random permutation of $\E$, and $k \sim \binom(n,p)$, it follows that $S$ includes each element of $\E$ independently with probability $p$. The set $R$ consists of all elements improving $S$ with respect to weight vector $w$, so $R \sim \imp(\M,p,w)$ by Definition~\ref{def:improving}. Since $L(e)=w(e)$ and $w \sim \mu$, it follows that $(R,L) \sim \impl(\M,p,\mu)$.

  We now condition on $S$ and $w$, which in turn fixes $(R,L)$, and show that $\pi'$ is a uniformly random permutation of $\E$. Since each iteration of the while loop feeds one of  $e_i$ or $e_j$ to $\phi$, and increments the corresponding counter ($i$ or $j$), it follows that $\pi'=(e'_1,\ldots,e'_n)$ is a permutation of $\E$.  Now consider the $t$th iteration of the while loop, let $S_t = S \sm \set{e'_1,\ldots,e'_{t-1}}$ and $\bar{S}_t = (\E \sm S) \sm \set{e'_1,\ldots,e'_{t-1}}$, and notice that $S_t \union \bar{S}_t = \E \sm \set{e'_1,\ldots,e'_{t-1}}$ is the set of elements not yet fed to $\phi$. It is easy to see inductively that $S_t = \set{e_i,\ldots,e_k}$ and $\bar{S}_t = \set{e_j,\ldots,e_n}$, where $i$ and $j$ are as in iteration $t$. Since $\pi$ is uniformly random, $e_i$ is a uniformly random element of $S_t$, and $e_j$ is a uniformly random element of $\bar{S}_t$. The bias of the coin in Step~\eqref{step:coin} is such that $e'_t= e_j$ with probability $\frac{|\bar{S}_t|}{|S_t \union \bar{S}_t|}$, and $e'_t=e_i$ with probability $\frac{|S_t|}{|S_t \union \bar{S}_t|}$. Therefore, $e'_t$ is a uniformly-random sample, without replacement, from $\bar{S}_t \union S_t = \E \sm \set{e'_1,\ldots,e'_{t-1}}$. It follows inductively that $\pi'$ is a uniformly random permutation of~$\E$.
\end{proof}
\begin{sublemma}\label{slem:phi_balanced}
  The RO-LCRM $\phi$ instantiated in Step~\eqref{step:lcrm} is $\beta$-balanced for $\impl(\M,p,\mu)$.
\end{sublemma}
\begin{proof}
  Follows directly from the fact that $\Phi$ is $(\beta,p)$-universal, and the fact that $\impl(\M,p,\mu)$ is $p$-uncontentious as shown in Sublemma~\ref{slem:labeled_improving_uncontentious}.
\end{proof}

Let $T \sse R$ denote the set of elements accepted by Algorithm~\ref{alg:psec_to_lcr}, as determined in Step~\eqref{step:feed_active}.
We can bound the expected weight of these elements as follows, where expectations are with respect to $w \sim \mu$,  $\pi \sim \E!$, the internal randomness in Algorithm~\ref{alg:psec_to_lcr}, and any randomness in the instantiated contention resolution map $\phi$. 
\begin{align*}
  \Ex [w(T)] &= \sum_{e \in \E} \sum_{w_0 \in \L} w_0 \cdot \Pr [e \in T \land w(e) = w_0] \\
             &= \sum_{e \in \E} \sum_{w_0 \in \L} w_0 \cdot \Pr [e \in T \land L(e) = w_0] \\
             &\geq \beta \sum_{e \in \E} \sum_{w_0 \in \L} w_0 \cdot \Pr [e \in R \land L(e) = w_0] &\mbox{(Sublemmata~\ref{slem:phi_input} and \ref{slem:phi_balanced})} \\
             &= \beta \sum_{e \in \E} \sum_{w_0 \in \L} w_0 \cdot \Pr [e \in R \land w(e) = w_0] \\
             &= \beta \Ex[ w(R)] \\
             &\geq \beta (1-p) \Ex[ \rank_w(\M)] &\mbox{(Fact~\ref{fact:improving_approximate} and Sublemma~\ref{slem:phi_input})} \\
             &= \beta (1-\alpha) \Ex[ \rank_w(\M)] \\  
\end{align*}
We conclude that Algorithm~\ref{alg:psec_to_lcr} is $\beta (1-\alpha)$ competitive for the prophet secretary problem on $\M$ with a finitely-supported prior. This concludes the proof of Lemma~\ref{lem:psec_to_lcr}.



%% file: lcr_to_cr.tex
\section{Reducing Labeled to Unlabeled Contention Resolution, Online}
\label{sec:lcr_to_cr}

Consider labeled contention resolution  for matroid $\M=(\E,\I)$ and labels $\L$ in the online random-arrival model, and denote $n= |\E|$ and $m=|\L|$. 
%
%
Here, a labeled set $(R,L)$ drawn from a known distribution is presented online to an LCRM for $\M$ and $\L$ as the string 
\begin{align}
  \label{eq:x}
  x=x(R,L,\pi)=(e_1,\ell_1),(e_2,\ell_2),\ldots,(e_n,\ell_n),
\end{align}
where $\pi=(e_1,\ldots,e_n)$ is a uniformly random permutation of $\E$, and $\ell_i \in \L \union \set{\bot}$ is the label $L(e_i)$ if $e_i \in R$ (i.e. $e_i$ is active) and is $\bot$ otherwise. Entries of $x$ are revealed online, with iteration $i$ revealing $(e_i,\ell_i)$, at which point the LCRM must immediately decide whether to accept $e_i$ in the event it is active. 

Recall from Section~\ref{sec:lcr} that, in the offline setting, labeled contention resolution on $\M$ and $\L$ reduces to unlabeled contention resolution on $\M \otimes \L$, via the map $(R,L) \to R \odot L$. It is therefore tempting to attempt a similar reduction in the online random order model as well. When the unlabeled problem on $\M \otimes \L$ is considered in the online random order model,  the (unlabeled) active set $R \odot L$ is presented online to an (unlabeled) CRM for $\M \otimes \L$ as the string 
\begin{align}
  \label{eq:y}
y=y(R \odot L ,\pi')= ((e'_1,\ell'_1),a_1), ((e'_2,\ell'_2),a_2), \ldots, ((e'_{nm}, \ell'_{nm}),a_{nm}),  
\end{align}
where  $\pi'=(e'_1,\ell'_1),\ldots,(e'_{nm},\ell'_{nm})$ is a uniformly random permutation of $\E \times \L$, and  $a_i \in \set{\top,\bot}$ designates whether $(e'_i,\ell'_{i}) \in R \odot L$. The string $y$ is revealed online, with iteration $i$ revealing $((e'_i,\ell'_i), a_i)$, at which point the CRM must immediately decide whether to accept $(e'_i,\ell'_i)$ in the event that $a_i=\top$. We emphasize that the string $y$ is longer than $x$: whereas an element $e \in \E$ appears exactly once in $x$, it appears $m$ times in $y$ (once for each possible label, with at most one of these appearances active). 

In attempting an online reduction from the labeled problem to its unlabeled counterpart, the problem we face at this point, intuitively, is the following: Given $x$, how do we ``interleave'' the ``missing'' element/label pairs to form the string $y$. This interleaving must be done online, before we know exactly which elements are active and what their labels are. Moreover, it must be such that the resulting order of element/label pairs in $y$ is uniformly distributed, at least approximately, in order to make use of any guarantee on the balance ratio of the (unlabeled) RO-CRM. This, it so happens, is nontrivial.

The reader might understandably furrow their brow at this point: \emph{Surely,  any ``reasonable'' random-order contention resolution algorithm need only exploit the relative ordering of active elements. This is already true in $x$, so an arbitrary interleaving of the missing element/label pairs should suffice! Certainly, this additional difficulty is an artifact of the precise technical definition of the random order model, rather than a conceptually interesting distinction!} The reader would be justified in expressing such skepticism. However, intuitive as it may seem, this knee-jerk reaction is flawed in a formal sense. Specifically, we show in Appendix \ref{app:onlyactive} that there does not exist a constant-competitive universal CRS in the online model where active elements arrive in a uniformly random order, but inactive elements are ordered arbitrarily. This impossibility result holds even for a $1$-uniform matroid. Therefore, for online contention resolution to plausibly encode the matroid secretary problem, it needs to exploit randomness in the arrival order of both active and inactive elements!

\subsection{Difficulties with Direct Approaches}
\label{sec:difficulty}

We begin by explaining how the direct approach, namely reducing the online labeled contention resolution  for $(\M,\L)$ to online unlabeled contention resolution for  $\M \otimes \L$, appears unlikely to succeed. Let $x=x(R,L,\pi)$ be the online input string to the labeled problem, as in Equation~\eqref{eq:x}. All online reductions to the corresponding unlabeled problem which are conceivable to us fit the following template, which has oracle access to an online CRM $\phi'$ for $\M \otimes \L$, and produces an online LCRM $\phi$ for $(\M,\L)$.

\begin{itemize}
\item While not all element/label pairs have been fed to $\phi'$, do one of the following:
  \begin{enumerate}[label=(\roman*)]
  \item Read the next active element/label pair $(e,\ell)$ in $x$ (if any), skipping inactive elements as needed. If $(e,\ell)$ has not previously been fed to $\phi'$, then feed $((e,\ell),\top)$ to $\phi'$, and accept $e$ iff $\phi'$ accepts $(e,\ell)$. 
  \item ``Hallucinate'' an element/label pair $(e,\ell)$ which has not yet been fed to $\phi'$, and feed $((e,\ell),\bot)$ to $\phi'$.
    \end{enumerate}
\end{itemize}

 

 Notice that, in each iteration, the choice to do (i) or (ii), and the choice of ``hallucination'' $(e,\ell)$ in (ii), can depend on previously observed entries of $x$, on previous acceptance/rejection decisions of $\phi'$, and on previous ``hallucinations''. These choices may also be randomized. Let $y$ denote the string fed to $\phi'$ through the course of the reduction, and let $\pi'$ denote the sequence of element/label pairs appearing in $y$.

For an instantiation of the above template to serve as an approximation preserving reduction (up to a constant) from the labeled problem to its unlabeled counterpart in the online random order model, the following properties appear needed.

\begin{enumerate}[label=(\alph*)]
\item Condition on the labeled set $(R,L)$, and assume that the order $\pi \in \E!$ of elements in $x$ is uniformly distributed (as is guaranteed by the random order model for the labeled problem). The order $\pi' \in (\E \times \L)!$ of element/label pairs in $y$ should be uniformly distributed (as is required by the random order model for the unlabeled problem) or approximately so (say, in terms of total variation distance).
\item In the event that $(e,\ell)$ is an entry of $x$ (i.e., $e$ is active with label $\ell$),  it should hold with constant probability that $((e,\ell), \top)$ is an entry of $y$ (i.e., $(e,\ell)$ is active in the corresponding unlabeled instance). This requires that $(e,\ell)$ is not ``hallucinated'' before it arrives in $x$.
\end{enumerate}

Trivial insantiations of our template satisfy one of (a) or (b), but satisfying both (a) and (b) simultaneously appears impossible.  To illustrate the difficulty, consider the special case where the number of active elements $|R|$ is known in advance. Arguably the most natural instantiation of our template in this special case, and one which at first glance appears promising, is as follows. 
In each iteration, with $r$ active entries of $x$ remaining and $k$ element/label pairs not yet fed to $\phi'$, we choose (i) with probability $p=p(r,k)=\frac{r}{k}$ and choose (ii) otherwise. When (ii) is chosen, we let $(e,\ell)$ be a uniformly random draw from the $k$ remaining element/label pairs. The probability $p$ is chosen to reflect the proportion of active to inactive element/label pairs.

It is not too difficult to verify that (b) is satisfied for this reduction. However, it can be shown that the permutation $\pi'$ is not uniformly distributed after conditioning on $(R,L)$. To see this, consider an element $e \in R$ with $L(e) = \ell$. The probability that $(e,\ell)$ is the first element/label pair appearing in $y$ is given by
\begin{align*}
  \frac{|R|}{mn} \cdot \frac{1}{|R|} + \frac{mn - |R|}{mn} \cdot \frac{1}{mn},
\end{align*}
where the first term corresponds to the event that (i) is chosen and $(e,\ell)$ is the first active element/label pair in $x$, and the second term corresponds to the event that (ii) is chosen and $(e,\ell)$ is hallucinated. Since $|R| \leq n$, this expression is at least $(2-\frac{1}{m}) \cdot \frac{1}{mn}$. When the number of labels $m$ is large, this is almost twice the probability that $(e,\ell)$ would appear first in a uniformly random permutation on element/label pairs! In other words, an active element/label pair is almost twice as likely to appear early in $\pi'$ than an inactive element/label pair, rendering $\pi'$ far from uniformly distributed. In fact, we can show that the total variation distance between $\pi'$ and the uniform distribution tends to $1$ as $m$ grows large, violating (a).

One might hope that different choices of $p(r,k)$ , coupled with a different rule for choosing the hallucinated element/label pair in (ii), might remedy this failure. However, some examination suggests that such approaches are unlikely to succeed. The difficulty, intuitively, is the following: when hallucinating inactive element/label pairs early in the sequence $y$, we must do so without knowledge of which active elements/label pairs appear later in $x$, and this is due to the online nature of the reduction. This gives active element/label pairs in $x$ a ``greater than fair'' shot at appearing early in the sequence $y$ (violating (a)), unless one is content with ``ignoring'' entries of $x$ with high probability (which results in violating (b)). Therefore, there is a tension between requirements (a) and (b).

These difficulties appear intrinsic to online reductions from the labeled problem on $(\M,\L)$ to the unlabeled problem on $\M \otimes \L$,  leaving little hope for preserving the balance ratio with such a direct approach. A new idea appears to be needed. 

\subsection{An Indirect Approach: Duplicating the labels}

We overcome these difficulties by reducing labeled contention resolution on $\M$ and $\L$ to unlabeled contention resolution on a much larger matroid than $\M \otimes \L$. Specifically, we  ``duplicate'' each label a large number of times, creating many ``identical copies'' of each element/label pair. We associate an active element/label pair in $x$ with one of its copies uniformly at random, leaving all other copies inactive. Roughly speaking, a random permutation $\pi'$ of the duplicated element/label pairs converges in probability to a limiting permutation as the number of copies grows large, modulo the symmetry between copies. An active element/label pair from $x$  is now merely a drop in a sea of its inactive brethren, and therefore interleaving $x$ into $\pi'$ has little influence on the probability distribution of $\pi'$.

Formally, we duplicate each label in $\L$ a large number $K$ of times to form an expanded set of labels $\L \times \C$,  where $\C$ is an abstract set for indexing copies with $|\C|=K$. We then reduce labeled contention resolution on $\M$ and $\L$  to unlabeled contention resolution on the matroid $\M \otimes (\L \times \C) = \M \otimes \L \otimes \C$. For $\ell \in \L$ and $c \in \C$, we say the pair $(\ell,c)$ is a \emph{copy} of label $\ell$. 
We also  say that an element $(e,\ell,c) \in \E \times \L \times \C$   of $\M \otimes \L \otimes \C$  is a \emph{copy} of $(e,\ell)$.

An offline version of our reduction maps an active set in $\M=(\E,\I)$ with labels in $\L$ to an (unlabeled) active set of $\M \otimes \L \otimes \C$ by selecting a copy of each label uniformly at random. Specifically, a labeled set of active elements $(R,L)$ is mapped to the (unlabeled) set  $R \odot L \odot C = \set{ (e,L(e),C(e)) : e \in R}$ of elements of the matroid $\M \otimes \L \otimes C$, where $C(e) \in \C$ is chosen independently and uniformly at random for each $e \in R$.\footnote{For convenience, we sometimes think of $C$ as a  function from $\E$ to $\C$, with the understanding that the restriction of $C$ to $R$, which we denote by $C|_R$, is all that is relevant for defining $R \odot L \odot C$.} It is easy to verify that if the random labeled set $(R,L)$ is $\alpha$-uncontentious (in the offline sense, of course), so is the random unlabeled set $R \odot L \odot C$.

\begin{observation}\label{obs:unc}
  If $(R,L)$ is an $\alpha$-uncontentious labeled set for $\M=(\E,\I)$ and $\L$, and $C: \E \to \C$ is chosen uniformly at random independently of $(R,L)$, then $R \odot L \odot C$ is an $\alpha$-uncontentious (unlabeled) set for $\M \otimes \L \otimes \C$.
\end{observation}
\begin{proof}
  First, we can interpret an $\alpha$-balanced offline LCRM for $(R,L)$ as an offline CRM for $R \odot L$ in the matroid $\M \otimes \L$. Second, we can interpret that latter as an offline CRM for $R \odot L \odot C$ in the matroid $\M \otimes \L \otimes \C$ --- in particular, one which ignores the index $C(e)$ for each element $(e,L(e),C(e))$. Since the indices $\set{C(e)}_{e \in \E}$ are independent of $R$ and $L$, it follows that that the balance ratio is preserved.
\end{proof}

In the online random order model, our reduction approximates  the above-described map $(R,L) \to R \odot L \odot C$. Even more importantly, if the elements $e \in \E$ are presented to our reduction in uniformly random order (each tagged with its label $L(e)$ if $e \in R$, or $\bot$ otherwise), then its output is (approximately) a uniformly random permutation of $\E \cross \L \cross \C$, with each $(e,\ell,c)$ tagged with $\top$ if $e \in R$, $L(e) = \ell$, and $C(e) = c$, and with $\bot$ otherwise. The error in both these approximations (the active set itself and the permutation), as measured in total variation distance, tends to~$0$ as the number of copies $K$ of each label grows large. The reduction is summarized in Algorithm~\ref{alg:lcr_reduction}.

\begin{alg}
  \PARAMETER Matroid $\M=(\E,\I)$ with $n$ elements, and a set $\L$ of $m$ labels.
  \PARAMETER Abstract index set $\C$ with $|\C|=K$
  \PARAMETER Oracle access to an online (unlabeled) CRM $\phi'$ for $M \otimes \L \otimes \C$
  \INPUT String $x=(e_1,\ell_1),\ldots,(e_n,\ell_n)$ given online, where $\pi=(e_1,\ldots,e_n)$ is a permutation of $\E$, and $\ell_i \in \L \union \set{\bot}$.
  \STATE Let $\pi'=\pi'(1), \ldots, \pi'(nmk)$ be a uniformly-random permutation of $\E \cross \L \cross \C$
  \STATE Draw $n$ integers i.i.d. from the uniform distribution on $[K]=\set{1, \ldots, K}$. Sort these integers in non-decreasing order $k_1 \leq k_2 \leq \ldots \leq k_n$.
  \STATE Let $i'=1$ be the current position in $\pi'$
  \FOR[Receive and process the $i$th online input $x_i=(e_i,\ell_i)$]{$i=1$  \TO $n$}
  \STATE Read the next online input $x_i = (e_i,\ell_i)$
  \IF[$e_i$ is inactive]{$\ell_i = \bot$} 
  \STATE Do nothing
  \ELSIF{There are at least $k_i$ copies of $(e_i,\ell_i)$ among $\pi'(1),\ldots,\pi'(i'-1)$}
  \STATE{\bf FAIL}
  \COMMENT{We already ``missed'' the $k_i$th copy of $(e_i,\ell_i)$}
  \ELSE[Skip ahead to the $k_i$th copy of $(e_i,\ell_i)$]
  \WHILE{$\pi'(i')$ is not the $k_i$th copy of $(e_i,\ell_i)$ seen so far in $\pi'$}  
  \STATE Feed $(\pi'(i'),\bot)$ to $\phi'$ as its next online input.
  \STATE Increment $i'$
  \ENDWHILE
  \COMMENT{$\pi'(i')$ is $k_i$th copy of $(e_i,\ell_i)$ in the ordered list $\pi'$}
  \STATE Feed $(\pi'(i'),\top)$ to $\phi'$ as its next online input, and {\bf ACCEPT} $e_i$ if $\phi'$ accepts $\pi'(i')$, otherwise {\bf REJECT} $e_i$.
  \STATE Increment $i'$
  \ENDIF
  \ENDFOR
  \WHILE[Complete the execution of $\phi'$ (may be omitted)]{$i' \leq nmk$}
  \STATE Feed $(\pi'(i'),\bot)$ to $\phi'$ as its next online input.
  \STATE Increment $i'$
  \ENDWHILE
  \caption{Reduction from labeled to unlabeled online contention resolution.}
  \label{alg:lcr_reduction}
\end{alg}

 Algorithm~\ref{alg:lcr_reduction} is an online LCRM for the matroid $\M$ with labels $\L$, which uses an online CRM $\phi'$ for $\M \otimes \L \otimes \C$ as a subroutine. At iteration $i$, the algorithm is presented with  $x_i=(e_i,\ell_i)$, where $\ell_i$ is either a label (if $e_i$ is active) or $\bot$ (if $e_i$ is inactive), and in the former case must decide ``on the spot'' whether to accept $e_i$. To guide these decisions, the algorithm runs a parallel execution of the CRM $\phi'$, and feeds the elements of $\M \otimes \L \otimes \C$ (with each labeled as active or inactive) to $\phi'$ in a uniformly random order $\pi'$.  For each active $e_i$ in the input string $x$, the algorithm (tries to) activate the $k_i$th copy of $(e_i,\ell_i)$ in the order of appearance in $\pi'$, where $\vec{k}=(k_1,\ldots,k_n)$ are $n$ i.i.d. uniform samples from $[K]$ ordered in non-decreasing order. The algorithm accepts $e_i$ if the corresponding activated copy of $(e_i,\ell_i)$ is accepted by $\phi'$.  To enable online acceptance/rejection decisions, we do the following: In each iteration  $i$ where $e_i$ is active, the algorithm ``skips ahead'' in $\pi'$ --- feeding skipped over elements as inactive to $\phi'$ --- until the desired $k_i$th copy of $(e_i,\ell_i)$ is reached, at which point this copy of $(e_i,\ell_i)$ is fed to $\phi'$ as active. The algorithm can fail when it ``skips over'' an element of $\M \otimes \L \otimes \C$ which we later realize should have been activated.

The following sequence of sublemmata lead to a proof of Lemma~\ref{lem:lcr_to_cr}.

\begin{sublemma}\label{slem:fail}
  For each input string $x$, the probability that Algorithm \ref{alg:lcr_reduction} {\bf FAIL}s tends to $0$ as $K \to \infty$. 
\end{sublemma}
\begin{proof}
  In iteration $i$, the algorithm skips through $\pi'$ until it finds the $k_i$th copy of $(e_i,\ell_i)$ in $\pi'$. It fails when that copy has already been passed over, in an earlier iteration $j<i$ while searching for the $k_{j}$th copy of $(e_{j}, \ell_{j})$. In particular, for the algorithm to fail it must be that there are $j<i$ such that at least $k_{i}$ copies of $(e_i,\ell_i)$ precede the $k_j$th copy $(e_j,\ell_j)$. We will show that this is a low-probability event.

  First, we show that $k_i - k_j \geq K^{0.75}$, simultaneously for all $1\leq j<i \leq n$, with high probability at least $1- \frac{n^2}{K^{0.25}}$. By definition, this is equivalent to showing that $n$ i.i.d. samples from the uniform distribution on $[K]$ are pairwise separated by at least $K^{0.75}$ with the claimed probability. The probability that the $(m+1)$st sample is at least $K^{0.75}$ away from the first  $m$ samples is at least $1-2m \frac{K^{0.75}}{K} = 1- 2m K^{-0.25}$, so we get
  \begin{align*}
    \Pr[\forall i \  k_{i+1} - k_i \geq K^{0.75}] &\geq \prod_{m=0}^{n-1} (1-2 m K^{-0.25})\\
                                                  &\geq 1- \sum_{m=0}^{n-1} 2 m K^{-0.25} \\
                                                  &= 1- 2K^{-0.25} \sum_{m=0}^{n-1} m \\
                                                  &\geq 1-n^2 K^{-0.25}
  \end{align*}

  Next, for $j<i$ we bound the probability that at least $k_j + K^{0.75}$ --- with high probability a lower bound on $k_i$ --- copies of $(e_i,\ell_i)$ precede the $k_j$th copy of $(e_j,\ell_j)$. Consider $K$ red balls and $K$ blue balls ordered uniformly at random, with red balls corresponding to copies of $(e_i,\ell_i)$ and blue balls corresponding to copies of $(e_j,\ell_j)$. It suffices to upperbound the probability that, for any prefix of the randomly ordered balls, the number of red balls exceeds the number of blue balls by more than $K^{0.75}$. For the first $m$ balls in the random order, we use the Hoeffding bound for sampling without replacement (see \cite{hoeffding}) to get a probability upperbound of $\exp(-\frac{2K^{1.5}}{m}) \leq \exp(-\frac{2K^{1.5}}{2K}) = \exp(-\sqrt{K})$. Taking the union bound over all $m=1,\ldots,2K$, we get a bound of $\frac{2K}{e^{\sqrt{k}}}$.

  Using the union bound, we conclude that the probability of failure is at most \[\frac{n^2}{K^{0.25}} + \sum_{j<i} \frac{2K}{e^{\sqrt{k}}} \leq n^2 \left(\frac{1}{K^{0.25}} + \frac{2K}{e^{\sqrt{K}}}\right),\]
  which tends to $0$ as $K \to \infty$.
\end{proof}

\begin{sublemma}\label{slem:lcr_to_cr2}
  Let $(R,L)$ be a random labeled set for an $n$-element matroid $\M=(\E,\I)$ and labels $\L$, let $\tilde{C}:\E \to \C$ be chosen uniformly at random, let $\pi$ be a uniformly random permutation of $\E$, and let $\tilde{\pi}$ be a uniformly random permutation of $\E \cross \L \cross \C$, with all four mutually independent.  Consider running Algorithm~\ref{alg:lcr_reduction} on the (random) input string $x= x(R,L,\pi)$ (see Equation~\eqref{eq:x}), and let $y'$ be the string of inputs passed to $\phi'$. After conditioning on Algorithm~\ref{alg:lcr_reduction} not {\bf FAIL}ing, the total variation distance between $y'$ and  $y(R \odot L \odot \tilde{C},\tilde{\pi})$ (see Equation~\ref{eq:y}) tends to $0$ as $K \to \infty$.
\end{sublemma}
\begin{proof}
  Let $x=x(R,L,\pi)=(e_1,\ell_1),\ldots(e_n,\ell_n)$, and recall that  $\ell_i = \bot$ if $e_i \not\in R$, and $\ell_i = L(e) \in \L$ if $e_i \in R$. When the algorithm succeeds, it feeds the elements $\E \times \L \times \C$ to $\phi'$ in the order $\pi'$, and for each $e \in R$ it designates precisely one copy of $(e,L(e))$ as active --- namely, the $k_{\pi^{-1}(e)}$th copy of $(e,L(e))$ appearing in $\pi'$. We use  $C'(e) \in \C$ to denote the index of this $k_{\pi^{-1}(e)}$th copy of $(e, L(e))$, and note that $C': R \to \C$ is a function that depends on $(R,L)$, $\pi'$, $\pi$, and $\vec{k}$. In summary, when the algorithm succeeds we have $y'=y(R \odot L \odot C' , \pi')$.

  We now condition on $(R,L)$ and $\pi'$, and show (conditionally) that $C'$ is a uniformly random function from $R$ to $\C$. Since $\pi$ is a uniformly random order on $\E$, it follows that the map $e \to k_{\pi^{-1}(e)} $ is a uniformly random perfect matching of $\E$ to $\set{k_i}_{i=1}^n$. Since $\set{k_i}_{i=1}^n$ consists of $n$ i.i.d. draws from $[K]$, we conclude that $(k_{\pi^{-1}(e)})_{e \in \E}$ are i.i.d. draws uniformly from $[K]$. In other words, for each element $e \in R$ we independently activate a copy of $(e,L(e))$  uniformly at random --- in particular the $k_{\pi^{-1}(e)}$th copy in order of appearance in $\pi'$, where $k_{\pi^{-1}(e)} \sim [K]$. It follows that $(C'(e))_{e \in R}$ are i.i.d. uniform draws from $\C$, as needed.

  Since $\pi'$ is a uniformly random permutation of $\E \cross \L \cross \C$ independent of $(R,L)$, and $C': R \to \C$ is uniformly random for each realization of $\pi'$ and $(R,L)$, it follows that $(R,L, C', \pi') \sim (R,L,\tilde{C}|_{R}, \tilde{\pi})$. Recall that $y'=y(R \odot L \odot C' , \pi')$ when the algorithm succeeds. Since the probability of failure tends to $0$ as $K \to \infty$ (Sublemma~\ref{slem:fail}), we conclude that the total variation distance between $y'$ and $y(R \odot L \odot \tilde{C}|_R , \tilde{\pi})= y(R \odot L \odot \tilde{C} , \tilde{\pi})$ tends to $0$ as $K \to \infty$, as needed.
  \end{proof}

\begin{sublemma}\label{slem:lcr_to_cr3}
  Let $(R,L)$ be a random labeled set for matroid $\M$ and labels $\L$, and let $C:\E \to \C$ be chosen uniformly at random independent of $(R,L)$. If $\phi'$ is a $\beta$-balanced random-order CRM for the random set $R \odot L \odot C$ of elements of the matroid $\M \otimes \L \otimes \C$, then Algorithm~\ref{alg:lcr_reduction} instantiated with $\phi'$ is a $\tilde{\beta}(K)$-balanced random-order LCRM for $(R,L)$, where $\tilde{\beta}(K)$ converges to $\beta$ as $K \to \infty$.
\end{sublemma}
\begin{proof}
  In the random order model, Algorithm~\ref{alg:lcr_reduction} applied to the random labeled set $(R,L)$ receives the string $x=x(R,L,\pi)=(e_1,\ell_1),\ldots,(e_n,\ell_n)$ as input, where $\pi$ is a uniform random order independent of $(R,L)$. Sublemmata \ref{slem:fail} and \ref{slem:lcr_to_cr2} imply that the input to the parallel execution of $\phi'$ tends to $y=y(R \odot L \odot C,\tilde{\pi})$ as $K \to \infty$, where $\tilde{\pi}$ is a uniformly random permutation of $\E \cross \L \cross \C$ independent of $(R,L)$ and $C$.

    Recall that  $e_i$ is accepted by the algorithm if and only if a copy of $(e_i,\ell_i)$ is accepted by the parallel execution of $\phi'$.  Let $S \sse R$ be the set of elements accepted by the algorithm. Similarly, let $S' \sse \E \cross \L \cross \C$ be the set of elements accepted by the parallel execution of $\phi'$. It follows that $e \in S$ and $L(e) = \ell$ if and only if $(e,\ell,c) \in S'$ for some $c \in \C$.
Since the input string to $\phi'$ tends to $y=y(R \odot L \odot C,\tilde{\pi})$ (in the sense of their total variation distance tending to $0$), and $\phi'$ is $\beta$-balanced for $R \odot L \odot C$ in the random order model, it follows that there is $\tilde{\beta}$ converging to $\beta$ such that
 \begin{align*}
   \Pr[(e,\ell,c) \in S'] &\geq \tilde{\beta} \Pr[(e,\ell,c) \in R \odot L \odot C] \\
                          &= \tilde{\beta} \Pr[e \in R \land L(e)= \ell \land C(e)=c] \\
                          &= \tilde{\beta} \Pr[e \in R \land L(e)= \ell] \cdot \frac{1}{K}.
\end{align*}
Now fix $e \in \E$ and $\ell \in \L$. Since different copies of $(e,\ell)$ are parallel in $\M \otimes \L \otimes \C$, and $\phi'$ accepts an independent set, it follows that the events $(e,\ell,c) \in S'$ are mutually exclusive. Therefore,
 \begin{align*}
   \Pr[e \in S \land L(e)= \ell] &= \Pr[ \exists c \in \C  \mbox{ s.t. }   (e,\ell,c) \in S']\\
                                 &= \sum_{c \in \C} \Pr[(e,\ell,c) \in S'] \\
                                 &\geq \sum_{c \in \C} \tilde{\beta} \Pr[e \in R \land L(e)= \ell] \cdot \frac{1}{K}\\
                                 &= \tilde{\beta} \Pr[e \in R \land L(e)= \ell],
 \end{align*}
as needed to show that the Algorithm~\ref{alg:lcr_reduction}, instantiated with $\phi'$, is a $\tilde{\beta}$-balanced LCRM for $(R,L)$ in the random order model.
\end{proof}

\noindent Lemma \ref{lem:lcr_to_cr} follows directly from Sublemma \ref{slem:lcr_to_cr3} and Observation \ref{obs:unc}.


%% file: conclusion.tex
\section{Conclusion}
In this paper, we built on our prior work in \cite{ucrs} to show  the matroid secretary problem equivalent to universal random-order contention resolution for matroids. It is worth noting that our result is information theoretic, pertaining to the power of online algorithms; i.e., we did not concern ourselves with computational efficiency of our reductions.\footnote{In fact, we did not even describe how distributions are represented as input to a contention resolution scheme --- a prerequisite for defining computational efficiency of such a scheme.}

Our result indicates that the main challenge of the matroid secretary conjecture is resolving contention in the presence of a particular form of positive correlation. Specifically, it suffices to resolve contention online for  uncontentious distributions, which admit the structure captured by the polyhedral characterization in Theorem~\ref{thm:characterize_uncontentious}. This structure --- which as noted in \cite{ucrs} is a natural generalization of the well-known \emph{matroid covering theorem}  --- might lend just enough tractability to enable progress on the conjecture.

Another conceptual takeaway from our result pertains to the importance of \emph{cardinal} information in the matroid secretary problem, as compared to just \emph{ordinal} information about the relative ordering of the weights. Ordinal algorithms for secretary problems were explored by \cite{secretary_ordinal,secretary_strong_ordinal}, though whether the \emph{ordinal matroid secretary problem} is fundamentally more difficult than its classical (cardinal) counterpart remains open. Whereas our result does not definitively answer this question, it does indicate that the ``hard part'' of the matroid secretary problem is fundamentally ordinal in nature. Indeed, contention resolution involves no weights at all, and the set of improving elements can be determined online using just ordinal information.\footnote{That said, our first component reduction in Section~\ref{sec:sec_to_psec} --- from the secretary problem to the prophet secretary problem --- does use cardinal information. This leaves open the possibility that the cardinal matroid secretary problem is strictly easier than its ordinal counterpart.}
 

%% file: appendix_normalize_discretize.tex
\section{Proof of Sublemma \ref{slem:normalize_discretize}}
\label{app:normalize_discretize}
  Consider the matroid secretary problem on matroid $\M=(\E,\I)$ and arbitrary (unknown) weights $w \in \RRp^\E$. Denote $r^*=\rank(\M)$, $n=|\E|$,  $v^*=\rank_w(\M)$, and let $T^* \in \I$ be a maximum-weight independent set (i.e., with $w(T^*)=v^*$). Consider the algorithm which, with probability $\frac{1}{2}$,   runs the $\frac{1}{e}$-competitive algorithm for the single-choice secretary problem with weights $w$, and otherwise runs the following reduction to a normalized and discretized instance on a restriction of $\M$.
  \begin{itemize}
  \item Sample roughly half the elements: Let $k \sim \binom\left(n,\frac{1}{2}\right)$, and observe the weights of the first $k$ elements $S$ in the arrival order $\pi$, without accepting any.
  \item Let $r=\rank^\M(S)$ and $v=\rank_w^\M(S)$ be rank and weighted rank, respectively, of the sample.
  \item Let $\bar{S}=\E \sm S$ be the remaining (unsampled) elements.    
  \item Define transformed weights for the unsampled elements $e \in \bar{S}$ as follows:  $\hat{w}_e=0$ if $w_e < \frac{v}{32r}$, otherwise $\hat{w}_e$ is the result of rounding down $\frac{w_e}{v}$ to the nearest power of $2$.
  \item To select an independent subset of the remaining elements $\bar{S}$, invoke a matroid secretary algorithm for the remaining matroid $\M | \bar{S}$ with weights $\hat{w}$.
  \end{itemize}

  It is clear that the elements in $\bar{S}$ arrive in uniformly random order after $S$. It is also clear that the transformed weights $\set{\hat{w}_e}_{e \in \bar{S}}$ can be computed online from the original weights~$\set{w_e}_{e \in \bar{S}}$, as well as the rank $r$ and weighted rank $v$ of the sample. It follows that, for each realization of the random sample $S$, this is indeed a valid reduction to the matroid secretary problem on $\M | \bar{S}$ and $\hat{w}$. The following relationship between the original and transformed weights is easy to see, and will be useful for the remainder of this proof.
  \begin{equation}
    \label{eq:ww}
   \frac{w_e}{2v} - \frac{1}{64r}=\frac{w_e - v/32r}{2v} \leq \hat{w}_e \leq \frac{w_e}{v} \mbox{ for all  $e \in \bar{S}$}
  \end{equation}

  Observe that if there is an element with weight exceeding $\frac{v^*}{16}$, then running single-choice secretary algorithm with probability $\frac{1}{2}$ guarantees that we obtain a competitive ratio of at least $\frac{1}{2} \cdot \frac{1}{e} \cdot \frac{1}{16} = \frac{1}{32 e} > \frac{1}{256}$. Therefore, we henceforth assume that  $w_e \leq \frac{v^*}{16}$ for all $e \in \E$ and analyze the above reduction.

We first show that $v$ is within a constant of $v^*$ with constant probability. It is immediate that $v$ is upper-bounded by $v^*$. The lower-bound follows from a series of elementary calculations, using the fact that each element of $T^*$, with total weight $w(T^*)=v^*$, is in  $S$ independently with probability~$\frac{1}{2}$.
  \begin{align*}
    \Pr[v< \frac{v^*}{4}]       &\leq \Pr\left[ w(S \intersect T^*) < \frac{v^*}{4}\right] &\mbox{(Since $v>w(S \intersect T^*)$ )} \\
                                &\leq \exp\left(-\frac{2 (v^*/4)^2}{\sum_{e \in T^*} w_e^2}\right) &\mbox{(Hoeffing's Inequality)} \\
                                &= \exp\left(-\frac{(v^*)^2}{8\sum_{e \in T^*} w_e^2}\right) \\
                                &\leq \exp\left(-\frac{(v^*)^2}{8(\max_{e \in T^*} w_e) (\sum_{e \in T^*} w_e)}\right) &\mbox{(Holder's inequality)} \\
                                &= \exp\left(-\frac{(v^*)^2}{8(\max_{e \in T^*} w_e) \cdot v^*}\right)  \\
                                &\leq \exp(-2) &\mbox{(Since $\max_e w_e \leq v^*/16$)}  \\
  \end{align*}
  It follows that $v$ is in $[\frac{v^*}{4},v^*]$ with probability at least $1-\frac{1}{e^2}$.

  An even simpler argument shows that $r$ is within a constant of $r^*$. By our assumption that $w_e \leq v^*/16$, it follows that $r^* \geq 16$. A simple application of the Hoeffding bound, akin to that above,  implies that $r$ is in $[\frac{r^*}{4},r^*]$ with probability at least $1- \exp\left(-\frac{2 (r^*/4)^2}{r^*}\right) = 1-\frac{1}{e^2}$.

Now denote $\bar{v}=\rank_w^\M(\bar{S})$ and $\bar{r}=\rank^\M(\bar{S})$, and observe that $v$ and $\bar{v}$ are identically distributed, and the same is true for $r$ and $\bar{r}$. It follows from the union bound that  $v$ and $\bar{v}$ are in $[\frac{v^*}{4},v^*]$, and moreover $r$ and $\bar{r}$ are in $[\frac{r^*}{4},r^*]$,  with probability at least $1-\frac{4}{e^2} > \frac{1}{4}$. In this event, symmetry implies that $v \geq \bar{v}$ with probability at least $\frac{1}{2}$. Therefore, the following hold with probability at least $\frac{1}{8}$:
\begin{equation}
  \label{eq:vv}
  \frac{v^*}{4} \leq \bar{v} \leq v \leq v^*
\end{equation}
and
\begin{equation}
  \label{eq:rr}
r,\bar{r} \in \left[\frac{r^*}{4}, r^* \right].
\end{equation}

We now condition on \eqref{eq:vv} and \eqref{eq:rr}, which hold with probability at least  $\frac{1}{8}$, and show that the matroid secretary instance $(\M | \bar{S}, \hat{w})$ is normalized and discretized, and moreover that our reduction to this instance is approximation preserving up to a constant.

Normalization follows easily from \eqref{eq:ww}, \eqref{eq:vv} and \eqref{eq:rr}:
\[  \rank_{\hat{w}}(\M| \bar{S}) \leq \frac{1}{v} \rank_{w}(\M| \bar{S}) 
                              = \frac{\bar{v}}{v}
                              \leq 1. \]
                            and
                           \begin{align*}
                             \rank_{\hat{w}}(\M| \bar{S})   &\geq \frac{1}{2v} \rank_w(\M | \bar{S}) - \frac{1}{64r} \rank(\M | \bar{S}) \\
                                                            &= \frac{\bar{v}}{2v}  - \frac{\bar{r}}{64r}  \\
                             &\geq \frac{v^*/4}{2v^*} - \frac{r^*}{64r^*/4} = \frac{1}{16}                             
\end{align*}

For discretization, recall that by definition each transformed weight $\hat{w}_e$ for $e \in \bar{S}$ is either zero or the result of rounding down $w_e/v$ to a power of $2$, for $v/32r \leq w_e \leq \rank_w(\M | \bar{S}) = \bar{v} \leq v$. It follows that a non-zero $\hat{w}_e$ is a power of $2$ between $1/64r$ and $1$. Since $r \leq r^* \leq 4 \bar{r}$, a non-zero $\hat{w}_e$ is a power of $2$ between  $\frac{1}{256\bar{r}}$ and $1$. 

For the approximation, consider any $c$-competitive solution  $\hat{T}$ for the instance  $(\M | \bar{S}, \hat{w})$, and let $T'$ be an optimal solution for the instance $(\M | \bar{S}, w)$. We can show that $\hat{T}$ is $\frac{c}{16}$-competitive for the original instance $(\M,w)$, using \eqref{eq:ww}, \eqref{eq:vv}, and \eqref{eq:rr}:
\begin{align*}
  w(\hat{T}) &\geq v \cdot \hat{w}(\hat{T}) \\
       &\geq v \cdot c \cdot \hat{w}(T') \\
       &\geq v \cdot c \cdot \left(\frac{w(T')}{2v} - \frac{|T'|}{64r}\right) \\
             &\geq v \cdot c \cdot \left(\frac{w(T')}{2v} - \frac{\bar{r}}{64r }\right) \\
               &\geq v \cdot c \cdot \left(\frac{w(T')}{2v} - \frac{1}{16 }\right) \\
             &=  c \cdot \left(\frac{\bar{v}}{2} - \frac{v}{16}\right) \\
             &\geq  c \cdot \left(\frac{v^*}{8} - \frac{v^*}{16}\right) \\
             &= \frac{c}{16} v^*
\end{align*}

Recall that we run the reduction (rather than the single-choice secretary algorithm) with probability $\frac{1}{2}$. Also recall that we conditioned on \eqref{eq:vv} and \eqref{eq:rr}, an event which holds with probability at least $\frac{1}{8}$. Therefore, the loss in the approximation ratio is no worse than $\frac{1}{16}  \times \frac{1}{8} \times \frac{1}{2} = \frac{1}{256}$.


%% file: appendix_onlyactive.tex
\section{Only Active Elements Arrive in Uniformly Random Order}
\label{app:onlyactive}

We now consider a semi-random model of online arrivals, where the relative order of active elements is uniformly random, but the order is otherwise arbitrary. We will show that there exists a $\beta$-uncontentious distribution for the 1-uniform matroid, where $\beta$ can be made arbitrarily close to $1$, admitting no constant-balanced CRM in this semi-random arrival model. In fact, we will show this to be true even when active elements arrive first (in uniformly random order), followed by all inactive elements in an arbitrary order.

Let $\eps>0$, and let $n$ and $m$ be integers. We will later choose these parameters to enable our impossibility result.  Let $\E_i$ be a class of $N_i=n^{m-i}$ elements for each $i=0,\ldots m$, and let $\M$ be the $1$-uniform matroid on $\E=\union_{i=0}^m \E_i$.  Denote $\delta= \frac{\eps^2}{n^m}$ and draw the set $R \sse \E$ of active elements as follows:
\begin{itemize}
\item Let $k$ be a draw from the geometric distribution with parameter $1-\delta$, and let $p=p(k) = \frac{\eps}{n^{m-i}}$
\item Let $R$ include each element of $\union_{i\leq k} \E_i$ independently with probability $p$.
\end{itemize}

In the subsequent analysis, for a quantity $x=x(\epsilon)$ we say $x \to y$,  if $\lim_{\eps \to 0} x = y$. We also say a probability $p=p(\eps)$ approaches $q$ if $\lim_{\eps \to 0} p \geq q$. We say an event holds with high probability if its probability approaches $1$.

We argue that $R$ is $\beta$-uncontentious, for $\beta \to 1$, by considering the following offline CRM: For $k$ in the above sampling procedure, accept an arbitrary element in $R \intersect \E_k$, if any.\footnote{Note that, in the offline model, we can assume without loss of generality that the CRM has access to $k$: it can simply sample the distribution $k | R$.} Observe that each element $e\in\E_i$ is accepted by the CRM with probability at least $\Pr[k=i] \cdot \frac{\eps}{N_i} \cdot (1-\eps/N_i)^{N_i-1} \geq \Pr[k=i]\frac{\eps}{e^\eps N_i}$, and is active (i.e., in $R$) with probability at most $\Pr[k=i]\frac{\eps}{N_i} +  \Pr[k>i]$.  Noting that $\Pr[k>i] \leq \delta \Pr[k=i]$ by definition of the geometric distribution, and bounding $N_i \leq n^m$, this yields $\beta = \frac{1}{e^\eps(1+\eps)}$ as needed. 

We also argue that any CRM which is $\alpha$-balanced for $R$ must, in the event that $k=i$ and at least one element of $\E_i$ is active, accept an element of $\E_i$ with conditional probability approaching $\alpha$. First, for an element $e \in E_i$, we show that $k=i$ with high conditional probability given $e$ is active.
\begin{align*}
  \Pr[k=i | e \in R] &= \frac{\Pr[k=i] \Pr [ e \in R | k=i]}{\Pr[e \in R]} \\
               &\geq \frac{\eps \Pr[k=i]/N_i}{\eps \Pr[k=i]/N_i + \Pr[k>i]}\\
               &\geq \frac{\eps \Pr[k=i]/N_i}{\eps\Pr[k=i]/N_i + \delta \Pr[k=i]} \\
               &= \frac{\eps}{\eps+\delta N_i} \\
               &\geq \frac{\eps}{\eps+\delta n^m}\\
               &= 1/(1+\eps) \to 1             
\end{align*}
It follows that $\Pr[e \mbox{ accepted} | e \in R, k=i] \geq \alpha'$ for some $\alpha' \to \alpha$. Note also that, when $e$ is active and $k=i$, there are no other active elements in $\E_i$ with high probability. In other words, the events $f \in R | k=i$ for elements $f \in \E_i$ tend to disjointness as $\eps \to 0$.  The initial claim follows.

Now fix an online CRM $\phi$ with balance ratio $\alpha$ for $R$ in our semi-random arrival model. Suppose that the active elements $R=R(k)$ are presented to $\phi$ in a uniformly random random order $\pi=(e_1,\ldots,e_{|R|})$, followed by all the inactive elements in an arbitrary order. Note that $\phi$ does not know $k$ a-priori, but can only glean information about it from observing $R$. In the case that $k=0$, $R$ is empty with probability $(1-\frac{\eps}{n^m})^{n^m} \approx 1-\eps$, and consists of a single element in $\E_0$ with probability $n^m \frac{\eps}{n^m} (1- \frac{\eps}{n^m})^{n^m - 1} \approx \eps(1-\eps)$. The argument in the previous paragraph implies, therefore, that $\phi$ must accept the first active element (in $\E_0$, if any) with probability approaching $\alpha$ when $k=0$.  Now consider the case of $k=1$: $R$ consists of  $\binom(n^m, \frac{\eps}{n^{m-1}}) \approx \frac{\eps}{n^{m-1}} n^m = \eps n$ elements of $\E_0$, and $\binom(n^{m-1},\frac{\eps}{n^{m-1}}) \approx O(1)$ elements of $\E_1$. More formally, if we choose $n= \omega(\frac{1}{\eps})$, Chernoff bounds imply that that $R$ consists of $\Omega(\epsilon n)$ elements of $\E_0$ and $O(1)$ elements of $\E_1$ with high probability. Therefore, with high probability the first element $e_1$ in the sequence will be in $\E_0$, and by our previous argument for the case of $k=0$ --- since the $\phi$ cannot distinguish between $k=0$ and $k=1$ at the beginning of the sequence --- it must be accepted with probability approaching $\alpha$. Moreover, by our previous paragraph if there is an active element in $\E_1$ then the first such element must be accepted by $\phi$ with probability approaching $\alpha$.

This pattern continues inductively. Consider the case $k=i$, for an arbitrary $i$. Let $e'_j$ be the first active element in $\E_j$ appearing in the online order, if any. With high probability, $e'_j$ exists for all $j<i$, though $e'_{i}$ may not (in the event there are no active elements in $\E_{i}$). Notice that the relative proportion of $\E_{j} \intersect R$ to $\E_{j+1} \intersect R$  is  $\Omega(n)$ with high probability (i.e., with probability approaching 1 as $\epsilon$ approaches 0).  We can therefore choose $m$ as an increasing function of $\frac{1}{\epsilon}$ such that, with high probability, $e'_j$ precedes $e'_{j+1}$ in $\pi$ simultaneously for all $j=1, \ldots, m-1$. Also notice that, for $j,\ell \leq i-1$, the (distribution of) the relative size of $\E_j \intersect R$ to that of $\E_\ell \intersect R$ is the same whether $k=i$ or $k=i - 1$; the principle of deferred decisions then implies that $\phi$ cannot distinguish $k=i$ from $k=i-1$ until it first encounters $e'_i$. It follows that $\phi$ accepts each of $e'_1,\ldots,e'_{i-1}$ with probability approaching $\alpha$ by induction. Moreover, as previously argued it must accept $e'_{i}$, in the event it exists, with probability approaching $\alpha$. Taking $i=m$, it follows that $\alpha = O(1/m)$. Since $m$ grows without bound, this proves that no absolute constant balance ratio is possible.